\documentclass[11pt,letterpaper]{article}
\usepackage{hdx}

\newcommand{\Adj}[1]{\mathrm{Adj}\parens{#1}}

\newcommand{\dT}{T}
\newcommand{\dH}{H}
\newcommand{\LD}{\mathsf{LocalDensifier}}

\newcommand{\SimpComp}{\calS}

\newcommand{\kface}{k\text{-}\mathsf{faces}}
\newcommand{\Face}[1]{#1\text{-}\mathsf{faces}}

\newcommand{\Link}{\mathsf{link}}
\newcommand{\Expansion}{\mathsf{TwoSidedGap}}
\newcommand{\GlExp}{\mathsf{GlobalExp}}
\newcommand{\LocExp}{\mathsf{LocalExp}}
\newcommand{\tensor}{\otimes}
\newcommand{\Spec}{\mathsf{Spectrum}}
\newcommand{\dir}{\mathcal{E}}

\newcommand{\im}{\mathsf{image}}
\newcommand{\finalHDX}{\mathcal{Q}}
\newcommand{\baseHDX}{\mathcal{B}}
\newcommand{\baseExp}{G}
\newcommand{\up}{\uparrow}
\newcommand{\down}{\downarrow}

\newcommand{\duwalk}{\finalHDX_k^{\down\up}}
\newcommand{\duwalksp}{\wt{\finalHDX}_k^{\down\up}}
\newcommand{\duwalkbase}{\baseHDX_{\down\up}}
\newcommand{\oproj}{P_{o}}
\newcommand{\orest}[1]{R_{o,#1}}
\newcommand{\iproj}{P_{I}}
\newcommand{\irest}{R_{I}}

\newcommand{\States}{\mathsf{States}}

\newcommand{\gap}{\mathsf{OneSidedGap}}
\renewcommand{\tilde}{\wt}
\renewcommand{\Tilde}{\wt}

\begin{document}

\title{High-Dimensional Expanders from Expanders}
\author{Siqi Liu\thanks{EECS Department, UC Berkeley.  \texttt{sliu18@berkeley.edu}. This research was supported in part by a Google Faculty Award and donations from the Ethereum Foundation and the Interchain Foundation.} \and Sidhanth Mohanty\thanks{EECS Department, UC Berkeley.   \texttt{sidhanthm@cs.berkeley.edu}.  Supported by NSF Grant CCF-1718695.} \and Elizabeth Yang\thanks{EECS Department, UC Berkeley.   \texttt{elizabeth\_yang@berkeley.edu}.  Supported by NSF Grant DGE 1752814.}}
\date{\today}
\maketitle

\begin{abstract}
    We present an elementary way to transform an expander graph into a simplicial complex where all high order random walks have a constant spectral gap, i.e., they converge rapidly to the stationary distribution.  As an upshot, we obtain new constructions, as well as a natural probabilistic model to sample constant degree high-dimensional expanders.

    In particular, we show that given an expander graph $G$, adding self loops to $G$ and taking the tensor product of the modified graph with a high-dimensional expander produces a new high-dimensional expander.  Our proof of rapid mixing of high order random walks is based on the decomposable Markov chains framework introduced by \cite{jerrum}.
\end{abstract}

\thispagestyle{empty}
\setcounter{page}{0}
\newpage
\thispagestyle{empty}
\tableofcontents
\setcounter{page}{0}
\thispagestyle{empty}
\newpage

\section{Introduction}
Expander graphs are graphs which are sparse, yet well-connected. They play important roles in applications such as the construction of pseudorandom generators and error-correcting codes \cite{expander-code}. Motivated by both purely theoretical questions, such as the topological overlapping problem, and applications in computer science, such as PCPs, a generalization of expansion to high dimensional complexes has recently emerged. We work with $d$--dimensional complexes, which not only have vertices and edges, but also hyperedges of $k$ vertices, for any $k \leq d + 1$.  Whereas in the one-dimensional world of graphs, the properties of edge expansion, spectral expansion and rapid mixing of random walks are equivalent, their generalization to   several different characterizations of ``expansion'' have been developed for these high--dimensional complexes. In particular, the high-dimensional extension of spectral expansion is simple to state, and implies rapid mixing of high order walks \cite{kaufman-oppenheim} and agreement expanders \cite{dinur-kaufman}. 

We construct bounded-degree high--dimensional expanders of all constant--sized dimensions, where the high order random walks have a constant spectral gap, and thus mix rapidly. We base our HDX's from existing $\dT$-regular one-dimensional constructions, which can be sampled readily from the space of all $\dT$-regular graphs. This endows a natural distribution from which we can sample HDX's of our construction as well.  After the first version of this paper was written, it was brought to our notice by a reviewer that the construction in this paper has been previously discussed in the community.  Nevertheless, a contribution of our work is a rigorous analysis of the expansion properties of this construction.

One sufficient, but not necessary criterion that implies rapid mixing is \emph{spectral}, which comes from the graph theoretic notion below.

\begin{definition}[Informal] \label{local-spectral-def}
    A $d$--dimensional $\lambda$--spectral expander is a $d$--dimensional simplicial complex (i.e. a hypergraph whose faces satisfy downward closure) such that 
    \begin{itemize}
        \item (Global Expansion) The vertices and edges (sets of two vertices) of the complex constitute a $\lambda$--spectral expander graph,
        \item (Local Expansion) For every hyperedge $E$ of size $\leq d-1$ in the hypergraph, the vertices and edges in the "neighborhood" of $E$ also constitute a $\lambda$--spectral expander. (The precise definition of "neighborhood" will be discussed later.)
    \end{itemize}
\end{definition}

Most known constructions of bounded-degree high--dimensional spectral expanders are heavily algebraic, rather than combinatorial or randomized. In contrast, there are a wealth of different constructions for bounded-degree (one-dimensional) expander graphs \cite{expander-survey}. Some of these are also algebraic, such as the famous LPS construction of Ramanujan graphs \cite{lps}, but there are also many simple, probabilistic constructions of expanders. In particular, Friedman's Theorem says that with high probability, random $d$-regular graphs are excellent expanders \cite{friedman}. 

Unfortunately, random $d$--dimensional hypergraphs with low degrees are not $d$--dimensional expander graphs. For a hypergraph with $n$ vertices, we need a roughly $n\left(\frac{\log{n}}{n}\right)^{1/d}$-degree Erdos-Renyi graph to make the neighborhood of every hyperedge of size $\leq d-1$ to be connected with high probability. While random low degree hypergraphs are not high--dimensional expanders, our construction provides simple probabilistic high--dimensional expanders of all dimensions.

\subsection{Summary of our results}

\paragraph{Construction.}  We construct an $\dH$--dimensional simplicial complex $\finalHDX$ on $n\cdot s$ vertices, from a graph $G$ of $n$ vertices and a (small) $\dH$--dimensional complete simplicial complex $\baseHDX$ on $s$ vertices.  To construct $\finalHDX$, we replace each vertex $v$ of $G$ with a copy of $\baseHDX$ which we denote $\baseHDX_v$. Denote the copy of a vertex $w\in \baseHDX$ in $\baseHDX_v$ by $(v,w)$.  The faces of $\finalHDX$ are chosen in the following way: for every face $\{w_1,w_2,\dots,w_k\}$ in $\baseHDX$, add it $\{(v_1,w_1),(v_2,w_2),\dots,(v_k,w_k)\}$ to the complex, where for some edge $e$ in $G$, the vertices $v_1,\dots v_k$ are each one of the endpoints of $e$; in particular there are $2$ choices for each $v_i$.   The main punchline of our work is that when $G$ is a (triangle-free) expander graph, the high order random walks on $\finalHDX$ mix rapidly.  Specifically, we prove:
\begin{theorem}[Main theorem, informal version of \pref{thm:mixing-main}]
    Suppose $G$ is a triangle-free expander graph with two-sided spectral gap $\rho$.  For every $k$ such that $1\le k < H$, there is a constant $C$ depending on $k,H,s,\rho$, but independent of $n$ such that the Markov transition matrix for the up-down walk on the $k$-faces of $\finalHDX$ has two-sided spectral gap $C$.
\end{theorem}

\paragraph{First attempt at proving rapid mixing of high order random walks.}  \cite{kaufman-mass}, which introduced the notions of up-down and down-up random walks, and subsequent works \cite{dinur-kaufman,kaufman-oppenheim,kaufman-oppenheim-construction,bases-sampling} developed and followed the ``local-to-global paradigm'' to prove rapid mixing of high order random walks.  In particular, each of these works would:
\begin{enumerate}[A.]
\item \label{step:links-small-eigen} Establish that all the links of a relevant simplicial complex have ``small'' second eigenvalue.
\item \label{step:local-exp-to-rapid-mix} Prove or cite a statement about how rapid mixing follows from small second eigenvalues of links (such as \pref{thm:informal-KO17}).
\end{enumerate}
Then, \pref{step:links-small-eigen} and \pref{step:local-exp-to-rapid-mix} together would imply that the up-down and down-up random walks on the simplicial complexes they cared about mixed rapidly.  This immediately motivates first bounding the second eigenvalue of the links of our construction, and applying the quantitatively strongest known version of the type of theorem alluded to in \pref{step:local-exp-to-rapid-mix}.  Thus, in \pref{sec:local-exp} we analyze the second eigenvalue of all links of $\finalHDX$ and prove:
\begin{theorem}[Informal version of \pref{thm:local-exp-main}]\label{thm:informal-local-exp} 
    The two-sided spectral gap of every link in $\finalHDX$ is bounded by approximately $\frac{1}{2}$.
\end{theorem}
And the `quantitatively strongest' known ``local-to-global'' theorem is
\begin{theorem}[Informal statement of {\cite[Theorem 5]{kaufman-oppenheim}}] \label{thm:informal-KO17}
    If the second eigenvalue of every link of a simplicial complex $\SimpComp$ is bounded by $\lambda$, then the up-down walk on $k$-faces of $\SimpComp$, $\SimpComp_k^{\up\down}$ satisfies:
    $$
        \lambda_2(\SimpComp_k^{\up\down}) \leq \left(1 - \frac{1}{k + 1} \right) + k\lambda.
    $$
\end{theorem}
Observe that the upper bound on the second eigenvalue of all links must be strictly less than $\frac{1}{k(k+1)}$ to conclude any meaningful bounds on the mixing time of the up-down random walk.  Thus, unfortunately, \pref{thm:informal-local-exp} in conjunction with \pref{thm:informal-KO17} fails to establish rapid mixing.

Hence, we depart from the local-to-global paradigm and draw on alternate techniques.

\paragraph{Decomposing Markov chains.}  Each $k$-face of $\finalHDX$ is either completely contained in a \emph{cluster} $\{(v,?)\}$ for a single vertex $v$ in $G$, or straddles two clusters corresponding to vertices connected by an edge, i.e., is contained in $\{(v,?)\} \cup \{(u,?)\}$.  Consider performing an up-down random walk on the space of $k$-faces of $\finalHDX$ (henceforth $\finalHDX_k^{\up\down}$).  If we record the single cluster or pair of clusters containing the $k$-face the random walk visits at each timestep, it would resemble:
\begin{displayquote}
    $\{17,19\}$ $\to$ $\{17,19\}$ $\to$ $\{17,19\}$ {\color{Fuchsia} $\to$ $\{17\}$} $\to$ $\{17\}$ {\color{Fuchsia}$\to$ $\{17, 155\}$} $\to$ $\{17, 155\}$ $\to$ $\{17, 155\}$ {\color{Fuchsia}$\to$ $\{155\}$} {\color{Fuchsia} $\to$ $\{155,203\}$} $\to$ $\{155,203\}$ $\to$ $\{155,203\}$ $\to$ $\{155,203\}$ $\to$ $\{155,203\}$ {\color{Fuchsia}$\to$ $\{203\}$ $\to$ $\{6,203\}$} $\to$ $\{6,203\}$ $\to$ $\cdots$
\end{displayquote}
In the above illustration of a random walk, let us restrict our attention to the segment of the walk where the $k$-faces are all contained in, say, the pair of clusters $\{155,203\}$.  Intuitively, we expect the random walk restricted to those $k$-faces to mix rapidly and also exit the set quickly by virtue of the state space being constant-sized.  In particular, if we keep the random walk running for $t \approx C\log n$ steps for some large constant $C$, it would seem that the number of ``exit events''\footnote{Transitions like {$\{17,19\}\to\{17\}$}, $\{155\}\to\{155,203\}$, and so on.} is roughly $\alpha\cdot C\log n$ for some other constant $\alpha$.  The sequence of ``exit events'' can be viewed as a random walk on the space of edges and vertices of $G$, and since there are \emph{many} steps in this walk, the expansion properties of $G$ tell us that the location of the random walk after $t$ steps is distributed according to a relevant stationary distribution.  In light of these intuitive observations of rapidly mixing in the walks within cluster pairs and also rapidly mixing in a walk on the space of cluster pairs, one would hope that the up-down walk on $k$-faces mixes rapidly.

This hope is indeed fulfilled and is made concrete in a framework of Jerrum et al. \cite{jerrum}.  In their framework, there is a Markov chain $M$ on state space $\Omega$.  They show that if $\Omega$ can be partitioned into $\Omega_1,\dots,\Omega_{\ell}$ such that the chain ``restricted'' (for some formal notion of restricted) to each $\Omega_i$, and an appropriately defined ``macro-chain'' (where each partition $\Omega_i$ is a state) each have a constant spectral gap, then the original Markov chain $M$ has a constant spectral gap as well.  Our proof of the fact that $\finalHDX^{\up\down}_k$ has a constant spectral gap utilizes this result of \cite{jerrum}.  This framework of decomposable Markov chains is detailed in \pref{sec:decompose}, and the analysis of the spectral gap of the down-up random walk\footnote{Which is actually equivalent to proving a spectral gap on the up-down random walk but is more technically convenient.  See \pref{fact:ud-du-same-spectra}.} is in \pref{sec:spec-gap-HO}.


\subsection{Related Work}

While high--dimensional expanders have been of relatively recent interest, already many different (non-equivalent) notions of high--dimensional expansion have emerged, for a variety of different applications.

The earliest notions of high--dimensional expansion were topological. In this vein of work, \cite{coboundary, coboundary-ii} introduced coboundary expansion, \cite{evra-kaufman} defined cosystolic expansion, and \cite{evra-kaufman, kkl} defined skeleton expansion. To our knowledge, most existing constructions of these types of expanders rely on the Ramanujan complex. We refer the reader to a survey by Lubotzky \cite{lubotzky-survey} for more details on these alternate notions of high dimensional expansion and their uses. 

To describe notions of high dimensional expansion that are relevant to computer scientists, we need to first highlight a key property of (one-dimensional) expander graphs--that random walks on them mix rapidly to their stationary distribution.  The notion of a random walk on graphs was generalized to simplicial complexes in the work of Kaufman and Mass \cite{kaufman-mass} to the ``up-down'' and ``down-up'' random walks, whose states are $k$-faces of a simplicial complex. They were interested in bounded--degree simplicial complexes where the up-down random walk mixed to its stationary distribution rapidly.  They then proceed to show that the known construction of Ramanujan complexes from \cite{lubotzky-ramanujan} indeed satisfy this property. 

A key technical insight in their work that the rapid mixing of up-down random walks follows from certain notions of \emph{local spectral expansion}, i.e., from sufficiently good \emph{two-sided} spectral expansion of the underlying graph of every link. A quantitative improvement between the relationship between the two-sided spectral expansion of links and rapid mixing of random walks was made in \cite{dinur-kaufman}, and this improvement was used to construct agreement expanders based on the Ramanujan complex construction.  Later, \cite{kaufman-oppenheim} showed that one-sided spectral expansion of links actually sufficed to derive rapid mixing of the up-down walk on $k$-faces. 



\subsubsection{HDX Constructions}

Although this combinatorial characterization of high--dimensional expansion is slightly weaker than some of the topological characterizations mentioned above, few constructions are known for \emph{bounded degree} HDX's with dimension $\geq 2$. Most of these rely on heavy algebra. In contrast, for one-dimensional expander graphs, there are a wealth of different constructions, including ones via graph products and randomized ones. \cite{friedman} states that even a random $d$-regular graph is an expander with high probability. 

The most well-known construction of bounded-degree high--dimensional expanders are the Ramanujan complexes \cite{lubotzky-ramanujan}. These require the \emph{Bruhat-Tits building}, which is a high-dimensional generalization of an infinite regular tree. The underlying graph has degree $q^{O(d^2)}$, where $q$ is a prime power satisfying $q \equiv 1 \pmod 4$. The links can be described by \emph{spherical buildings}, which are complexes derived from subspaces of a vector space, and are excellent expanders.
    
Dinur and Kaufman showed that given any $\lambda \in (0,1)$, and any dimension $d$, the $d$--skeleton of any $d + \ceil{2/\lambda}$--dimensional Ramanujan complex is a $d$--dimensional $\lambda$--spectral expander \cite{dinur-kaufman}. Here, the degree of each vertex is $\left(2/\lambda\right)^{O((d + 2/\lambda )^2)}$. In other words, they ``truncate'' the Ramanujan complexes, throwing out all faces of size greater than some number $k$. Their primary motivation was to obtain agreement expanders, which find uses towards PCPs.
    
Recently, Kaufman and Oppenheim \cite{kaufman-oppenheim-construction} present a construction of one--sided high--dimensional expanders, which are coset complexes of elementary matrix groups. The construction guarantees that for any $\lambda \in (0,1)$ and any dimension $d$, there exists a infinite family of high--dimensional expanders $\{X_i\}_{i\in\mathbb{N}}$, such that (1) every $X_i$ are $d$--dimensional $\lambda$--one--sided--expander; (2) every $X_i$'s 1-skeleton has degree at most $\Theta\left(\sqrt{\frac{(1/\lambda + d -1 )^{(d+2)^2}}{2\log{(1/\lambda+d-1)}}}\right)$; (3) as $i$ goes to infinity the number of vertices in $X_i$ also goes to infinity.

Even more recently, Chapman, Linial, and Peled \cite{chapman-linial-peled} also provided a combinatorial construction of two-dimensional expanders. They construct an infinite family of $(a, b)$-regular graphs, which are $a$-regular graphs whose links with respect to single vertices are $b$-regular. The primary motivation for their construction comes from the theory of PCPs. They prove an Alon-Boppana type bound on $\lambda_2(G)$ for any $(a, b)$-regular graph, and construct a family of graphs where this bound is tight. They also build an $(a, b)$-regular two-dimensional expander using any non-bipartite graph $G$ of sufficiently high girth; they achieve a local expansion only depending on the girth, and the global expansion depending on the spectral gap of $G$. Like ours, their construction also resembles existing graph product constructions of one-dimensional expanders.

\section{Preliminaries and Notation}

\subsection{Spectral Graph Theory}

While we can describe our constructions combinatorially, our analysis of both the mixing times of certain walks as well as the local expansion will heavily rely on understanding graph spectra. 

\begin{definition}
    For an edge-weighted directed graph $G$ on $n$ vertices, we use $\Adj{G}$ to denote its \emph{(normalized) adjacency matrix}, i.e. the matrix given by
    \[
        \Adj{G}_{(u, v)} = \frac{\bone_{(u, v)\in E(G)} \cdot w((u, v))}{\sum_{v: (u, v) \in E(G)} w((u, v))}
    \]
    and write its (right) eigenvalues as
    \[
        1 = \lambda_1(G) \ge \lambda_2(G) \ge \dots \ge \lambda_n(G) \ge -1
    \]
    Let $\Spec(G)$ to indicate the set $\{\lambda_i(G)\}$. We write $\gap(G)$ for the \emph{spectral gap} of $G$, which is the quantity $1 - \lambda_2(G)$. Graphs with $\gap(G) \geq \mu$ are \emph{one-sided $\mu$-expanders}.
    
    Most of the graphs we analyze achieve a stronger condition; that the second largest eigenvalue magnitude is not too large. Formally, we write $|\lambda|_i$ for the $i$-th largest eigenvalue in absolute value. In particular, $|\lambda|_2 = \max\{|\lambda_2|,|\lambda_n|\}$.  The \emph{absolute spectral gap} of $G$, denoted $\Expansion(G)$, is the quantity $1 - |\lambda|_2$.  Graphs with $\Expansion(G) \geq \mu$ are \emph{two-sided $\mu$-expanders}.
\end{definition}

\begin{remark}
    For an \emph{undirected} weighted graph, we simply have $w((u, v)) = w((v, u))$, and use this to define the adjacency matrix the same way. 
\end{remark}



\subsubsection{Graph Tensors}

\par \noindent Our construction can roughly be described as a tensor product, defined below. 
\begin{definition}
The tensor product $G \times H$ of two graphs $G$ and $H$ is given by
\begin{enumerate}
    \item Vertex set $V(G \times H) = V(G)\times V(H)$,
    \item Edge set $E(G \times H) = \{((u_1, v_1), (u_2, v_2)) : (u_1, u_2) \in E(G) \text{ and } (v_1, v_2)\in E(H) \}$.
\end{enumerate}
\end{definition}

\noindent The adjacency matrix $\Adj{G \times H}$ is the tensor (Kronecker) product $\Adj{G} \otimes \Adj{H}$. Due to this structure, $\Spec(G \times H) = \{\lambda_i\mu_j: \lambda_i \in \Spec{G}, \; \mu_j \in \Spec(H)\}$. As $1$ is the largest eigenvalue of both $\Adj{G}$ and $\Adj{H}$, it follows that both
\begin{align*}
\gap(G \times H) &= \min(1 - 1 \cdot \mu_2, 1 - \lambda_2 \cdot 1) = \min(\gap(G), \gap(H)) \\
\Expansion(G \times H) &= \min(1 - 1 \cdot |\mu|_2, 1 - |\lambda|_2 \cdot 1) = \min(\Expansion(G), \Expansion(H))
\end{align*}

\subsection{Markov Chains}

We provide a basic overview of the Markov chain concepts used to analyze our high order walks. We refer to \cite{markov-book} for a detailed and thorough treatment of the fundamentals of Markov chains.

\begin{definition}
A \emph{Markov chain} $M = (\Omega, P)$ is given by \emph{states} $\Omega$ and a \emph{transition matrix} $P$ where $P[i, j]$ is the probability of going to state $j$ from state $i$. We may also write this quantity as $M[j \to i]$. 
\end{definition}

\begin{remark}
The literature often defines $P_{i, j}$ as the probability $\Pr(i \to j)$, so their $P$ is the transpose of ours. However, we work with column (right) eigenvectors to analyze the spectrum of $P$, while this alternate convention uses row (left) eigenvectors, so both conventions yield the same results.
\end{remark}

\begin{definition}
    We can view any Markov chain $M$ as a weighted, directed graph $G$, defined by $V(G) = \States(M)$, $E(G) \coloneqq \{(i,j) : i,j\in V(G), M[i\to j] > 0\}$, and $w((i, j)) = M[j \to i]$. \\
    \\
    The transition matrix of $M$ is $\Adj{G}$, and we also refer to $\Spec(G)$ as the \emph{spectrum} of $M$. Every adjacency matrix has $\lambda_1 = 1$, so transition matrix of $M$ has an eigenvector $\pi_M$ (normalized so that entries sum to $1$) for the eigenvalue $1$. We call $\pi_M$ a \emph{stationary distribution} of $M$.
\end{definition}

\begin{remark}
We may use the term ``graph'' in lieu of ``chain'' when we want to indicate the random walk defined by the transition matrix $\Adj{G}$.
\end{remark}





\noindent The next property we introduce is present for every Markov chain we analyze. 
\begin{definition}
The Markov chain $M = (\Omega, P)$ is \emph{time-reversible} if for any integer $k \geq 1$:
$$
\pi_M(x_0)M[x_0 \to x_1] \cdots M[x_{k - 1} \to x_k] = \pi_M(x_k)M[x_k \to x_{k - 1}] \cdots M[x_1 \to x_0]
$$
\end{definition}

\noindent Intuitively, it means that if start at the stationary distribution and run the chain for a sequence of time states, the reverse sequence has the same probability of occurring. Time reversibility helps us compute stationary distributions via the \emph{detailed balance equations}. (This is especially helpful when there are a huge number of symmetric states.)
\begin{fact}
The Markov chain $M = (\Omega, P)$ is time-reversible if and only if it satisfies the \emph{detailed balance equations}: for all $x, y \in \Omega$,
$$
\pi_M(x) M[x \to y] = \pi_M(y) M[y \to x]
$$
\end{fact}


\begin{definition}
    The $\eps$-\emph{mixing time} of a Markov chain $M$ is the smallest $t$ such that for any distribution $\nu$ over the states of $M$,
    \[
        \|\pi_M-P^t\nu\|_1\le\eps
    \]
    where $\pi_M$ is the stationary distribution of $M$.
\end{definition}

\begin{theorem} \label{thm:spec-gap-to-mixing}
    For any Markov chain $M$, the $\eps$-mixing time $t_{\mathrm{mix}}(\eps)$ satisfies:
    \[
        t_{\mathrm{mix}}(\eps) \le \log\left(\frac{1}{\eps \min\pi_M}\right)\cdot\frac{1}{\Expansion(M)}.
    \]
\end{theorem}

\subsubsection{Decomposing Markov Chains}\label{sec:decompose}

Consider a finite-state time reversible Markov chain $M$ whose structure gives rise to natural state-space partition, $M$ can be decomposed into a number of restriction chains and a projection chain. \cite{jerrum} show that the spectral gap for the original chain can be lower bounded in terms of the spectral gaps for the restriction and projection chains.

We now formally define the decomposition of a Markov chain. Consider an ergodic Markov chain on finite state space $\Omega$ with transition probability $P\colon \Omega^2\rightarrow [0,1]$. Let $\pi \colon \Omega\rightarrow[0,1]$ denote its stationary distribution, and let $\{\Omega_i\}_{i\in[m]}$ be a partition of the state space into $m$ disjoint sets, where $[m]\coloneqq \{1,\dots,m\}$. 

The \emph{projection chain} induced by the partition $\{\Omega_i\}$ has state space $[m]$ and transitions
\begin{equation*}
    \overline{P}(i,j) = \left(\sum_{x\in\Omega_i} \pi(x)\right)^{-1}\sum_{x\in\Omega_i\\ y\in\Omega_j} \pi(x)P(x,y).
\end{equation*}
The above expression corresponds to the probability of moving from any state in $\Omega_i$ to any state in $\Omega_j$ in the original Markov chain.

For each $i\in[m]$, the \emph{restriction chain} induced by $\Omega_i$ has state space $\Omega_i$ and transitions
\begin{equation*}
    P_i(x,y) = \begin{cases} P(x,y), \quad & x\neq y, \\
    1- \sum_{z\in\Omega_i\setminus\{x\}} P(x,z), & x=y.
    \end{cases}
\end{equation*}
$P_i(x,y)$ is the probability of moving from state $x\in\Omega_i$ to state $y$ when leaving $\Omega_i$ is not allowed.

Regardless of how we define the projection and restriction chains for a time reversible Markov chain, they all inherit one useful property from the original chain. 

\begin{fact} \label{fact:inherit-time-reversibility}
    Let $M = (\Omega, P)$ be a time-reversible Markov chain. Then, for any decomposition of $M$, the projection and restriction chains are also time-reversible. 
\end{fact}

We ultimately want to study the spectral gap of random walks. Luckily, the original Markov chain's spectral gap is related to the restriction and projection chains' spectral gaps in the following way:

\begin{theorem}[{\cite[Theorem 1]{jerrum}}]\label{thm:Jerrum-et-al}
    Consider a finite-state time-reversible Markov chain decomposed into a projection chain and $m$ restriction chains as above. Define $\gamma$ to be maximum probability in the Markov chain that some state leaves its partition block,
    \begin{equation*}
        \gamma \coloneqq \max_{i\in[m]} \max_{x\in\Omega_i} \sum_{y\in\Omega\setminus\Omega_i} P(x,y).
    \end{equation*}
    Suppose the projection chain satisfies a Poincar\'{e} inequality with constant $\Bar{\lambda}$ , and the restriction chains satisfy inequalities with uniform constant $\lambda_{\min}$.
    Then the original Markov chain satisfies a Poincar\'{e} inequality with constant 
    \begin{equation*}
        \lambda \coloneqq \min\left\{\frac{\Bar{\lambda}}{3}, \frac{\Bar{\lambda}\lambda_{\min}}{3\gamma + \Bar{\lambda}}\right\}.
    \end{equation*}
\end{theorem}
Recall that if $\lambda$ satisfies a Poincar\'{e} inequality, it is a lower bound on the spectral gap (cf.\ \cite{markov-book}). 

\subsection{High-Dimensional Expanders}

The generalization from expander graphs to hypergraphs (more specifically, simplicial complexes) requires great care. We now formally establish the high dimensional notions of ``neighborhood'', ``expansion,'' and ``random walk.''

\begin{definition}
    A \emph{simplicial complex} $\SimpComp$ is specified by vertex set $V(\SimpComp)$ and a collection $\calF(\SimpComp)$ of subsets of $V(\SimpComp)$, known as \emph{faces}, that satisfy the ``downward closure'' property: if $A\in\calF(\SimpComp)$ and $B\subseteq A$, then $B\in\calF(\SimpComp)$. Any face $S\in\calF(\SimpComp)$ of cardinality $(k + 1)$ is called a \emph{$k$-face} of $\SimpComp$. We use $\kface(\SimpComp)$ to denote the subcollection of $k$-faces of $\calF(\SimpComp)$. We say $\SimpComp$ has \emph{dimension} $d$, where $d = \max\{|F|: F \in \calF(\SimpComp)\} - 1$.
\end{definition}



\begin{example}
    A $1$-dimensional complex $\calS$ is a graph with vertex set $V(\calS)$ and edge set $\Face{1}(\SimpComp)$.
\end{example}

\begin{definition}
To formally define random walks and Markov chains on a $\SimpComp$, we need to associate $\SimpComp$ with a weight function $w: \calF(\SimpComp) \to \mathbb{R}_+$. We want our weight function to be \emph{balanced}, meaning for $F \in \kface(\SimpComp)$:
$$
w(F) = \sum_{J \in (k + 1)\text{-}\mathsf{faces}(\SimpComp): J \supset F} w(J)
$$
If we restrict ourselves to balanced $w$, it suffices to only define $w$ over $d\text{-}\mathsf{faces}(\SimpComp)$ and propagate the weights downward to the lower order faces. 
\end{definition}



\begin{definition}
    The \emph{(weighted) $k$-skeleton} of $\SimpComp$ is the complex with vertex set $V(\SimpComp)$ and all faces in $\calF(\SimpComp)$ of cardinality at most $k + 1$, with weights inherited from $\SimpComp$.
\end{definition}

\begin{example}
    The $1$-skeleton of $\SimpComp$ only contains its vertices ($0$-faces) and edges ($1$-faces). It can be characterized as a graph with edge weights, so we can also compute $\gap(1\text{-}\mathsf{skeleton}(\SimpComp))$ and $\Expansion(1\text{-}\mathsf{skeleton}(\SimpComp))$. 
\end{example}

\begin{definition}
    For $S\in\kface(\SimpComp)$ for $k\le H-1$, we associate a particular $(H-k)$-dimensional complex known as the \emph{link} of $S$ defined below.
    \[
        \Link(S) := \{ T\setminus S: T\in\calF(\SimpComp), S\subseteq T \}
    \]
    If $\SimpComp$ was equipped with weight function $w$, then $\Link(S)$ ``inherits'' it. We associate $\Link(S)$ with weight function $w_S$ given by $w_S(T) = w(S \cup T)$. If $w$ is balanced, then $w_S$ is also balanced.  We call a $\Link(S)$ a $t$-link if $|S|$ has cardinality $t$.
\end{definition}

\begin{example}
    In a graph, the link of a vertex is simply its neighborhood.
\end{example}

\begin{definition}
    The \emph{global expansion} of $\SimpComp$, denoted $\GlExp(\SimpComp)$, is the expansion of its weighted $1$-skeleton.
\end{definition}

\begin{definition}
    The \emph{local expansion} of $\SimpComp$, denoted $\LocExp(\SimpComp)$ is
    \[
        \LocExp(\SimpComp) := \min_{0 \le k\le H-1} \min_{S \in \kface(\SimpComp)}\Expansion(1\text{-}\mathsf{skeleton}(\Link(S))).
    \]
    In words, it is equal to the expansion of the worst expanding link.
\end{definition}


\begin{example} \label{eg:complete-complex}
    We use $\calK_{H+2}^{(H)}$ to denote the \emph{complete $H$-dimensional complex} on vertex set $[H+2]$, i.e., the pure $H$-dimensional simplicial complex obtained by making the set of $(H+1)$-faces equal to all subsets of $[H+2]$ of size $H+1$.  The $1$-skeleton is then a clique on $H+2$ vertices whose expansion is $1-\frac{1}{H+1}$ and the $1$-skeleton of a $t$-link is a clique on $H+2-t$ vertices, which has expansion $1-\frac{1}{H+1-(t-1)}$.  As a result, $\Expansion\left(\calK_{H+2}^{(H)}\right) = \frac{1}{2}$.
\end{example}


\begin{remark}
    We often use $\Adj{\SimpComp}$ to refer to the adjacency matrix of the $1$-skeleton of $\SimpComp$, and we may also use $\lambda_i(\SimpComp)$ to refer to the $i$-th largest eigenvalue of $\Adj{\SimpComp}$.
\end{remark}

\noindent Previously, we mentioned that there are several different notions of high dimensional expansion: some geometric or topological, some combinatorial. We now formally define high dimensional spectral expansion, which is a more combinatorial and graph theoretic notion:
\begin{definition}
    $\SimpComp$ is a \emph{two-sided $\lambda$-local spectral expander} if $\GlExp(\SimpComp) \geq \lambda$ and $\LocExp(\SimpComp) \geq \lambda$.
\end{definition}

\subsubsection{High Order Walks on Simplicial Complexes}

Let $\SimpComp$ be a $H$-dimensional simplicial complex and with weight function $w:\kface(\SimpComp)\rightarrow\R_{\ge 0}$ on the $k$-faces of $\SimpComp$, for $k \le H$. For each $k < H$, we can define a natural (periodic) Markov chain on a state space consisting of $k$-faces and $(k + 1)$-faces of $\SimpComp$. 
\begin{itemize}
    \item At a $(k + 1)$-face $J$, there are exactly $(k + 2)$ faces $F \in \kface(\SimpComp)$ such that $F \subset J$, due to the downward closure property. We transition from $J$ to each $k$-face $F$ with probability $\frac{1}{k + 2}$. 
    
    \item At a $k$-face $F$, we transition to each $(k + 1)$-face $J$ satisfying $J \supset F$ with probability $\frac{w(J)}{w(F)}$. (Note that $w$ must be balanced for these transitions to be well-defined.)
\end{itemize}
Restricting the above chain to only odd or even time steps gives us two new random walks: one entirely on $\kface(\SimpComp)$ and one entirely on $(k + 1)\text{-}\mathsf{faces}(\SimpComp)$.

\begin{definition}[Down-up walk on $k$-faces of $\SimpComp$]\label{def:du-walk}
=    Let $\SimpComp_{k + 1}^{\down\up}$ be the Markov chain with state space equal to $\kface(\SimpComp)$ and transition probabilities $\SimpComp_{\down\up}[J \to J']$ described by the process above, where there is an implicit transition \emph{down} to a $k$-face and back \emph{up} to a $(k + 1)$-face. Then:
    \begin{align*}
        \SimpComp_{\down\up}[J \to J'] &=
        \begin{cases} \displaystyle\frac{1}{k + 1} \sum_{F \in \kface(\SimpComp): F \subset J} \frac{w(J)}{w(F)} & \text{ if } J = J' \\
        \displaystyle\frac{1}{k + 1} \cdot \frac{w(J')}{w(J \cap J')} & \text{ if } J \cap J' \in \kface(\SimpComp)  \\
        0 & \text{ otherwise}
        \end{cases}
    \end{align*}
\end{definition}

\begin{definition}[Up-down walk on $k$-faces of $\SimpComp$]\label{def:du-walk}
    Let $\SimpComp_{\up\down}$ be the Markov chain with state space equal to $\kface(\SimpComp)$ and transition probabilities $\SimpComp_{\up\down}[F\to F']$ described by the process above, where there is an implicit transition \emph{up} to a $(k + 1)$-face and back \emph{down} to a $k$-face. Then:
    \begin{align*}
        \SimpComp_{\down\up}[F\to F'] &= 
        \begin{cases}
        \displaystyle\frac{1}{k + 1}  & \text{ if } F = F' \\
        \displaystyle\frac{w(F \cup F')}{w(F)} & \text{ if } F \cup F' \in (k + 1)\text{-}\mathsf{faces}(\SimpComp)  \\
        0 & \text{ otherwise}
        \end{cases}
    \end{align*}
\end{definition}

\begin{remark}
    In the literature, we also see $\SimpComp_{k + 1}^{\down\up}$ written as $\SimpComp_{k + 1}^{\lor}$, and $\SimpComp_k^{\up\down}$ written as $\SimpComp_k^{\land}$.
\end{remark}

\noindent We now present some facts about these high order walks without proof. We refer to \cite{kaufman-oppenheim, bases-sampling} for proofs of these facts.

\begin{fact}    \label{fact:ud-du-same-spectra}
    The transition matrices for $\SimpComp_{k + 1}^{\down\up}$ and $\SimpComp_k^{\up\down}$ share the same eigenvalues. The nonzero eigenvalues occur with the same multiplicity. A straightforward but important consequence of this fact is
    $$
        \Spec{(\SimpComp_{k + 1}^{\down\up})} = \Spec{(\SimpComp_{k + 1}^{\up\down})}
    $$
\end{fact}

\begin{fact}
The Markov chains $\SimpComp_k^{\down\up}$ and $\SimpComp_k^{\up\down}$ have the same stationary distribution on $\kface(\SimpComp)$, which is proportional to $w(F)$ for each $F \in \kface(\SimpComp)$. We will call this distribution $\pi_k(\cdot)$. 
\end{fact}

For the remainder of the paper, we will assume a uniform weight function on $d\text{-}\mathsf{faces}(\SimpComp)$, which is useful for applications like sampling bases of a matroid \cite{bases-sampling}. When using the uniform weighting scheme, for $F \in \kface(\SimpComp)$, there is a natural interpretation of $\pi_k(F)$: the fraction of $d$-faces that contain $F$ as a subface. (We also note that we will use symbolic variables to represent various weight values, and that it is straightforward to adapt our computations to cases where we have uniform weights over $\kface(\SimpComp)$ for any $k$.)

\section{Local Densification of Expanders}

For a graph $G$ and $H$-dimensional simplicial complex $\SimpComp$, we give a way to combine the two to produce a bounded-degree $H$-dimensional complex $\LD(G,\SimpComp)$ of constant expansion. First, construct a graph $G'$ with
\begin{enumerate}
    \item vertex set equal to $V(G)\times V(\SimpComp)$, and
    \item edge set equal to $\{\{(v_1,b_1),(v_2,b_2)\} : \{b_1,b_2\}\in \Face{1}(\SimpComp), \{v_1,v_2\} \in E(G)\text{ or }v_1=v_2 \}$.
\end{enumerate}

$\LD(G,\SimpComp)$ is then defined as the $H$-dimensional pure complex whose $H$-faces are all cliques on $H+1$ vertices $\{(v_1,b_1),(v_2,b_2),\dots,(v_{H+1},b_{H+1})\}$ such that there exists an edge $\{a,b\}$ in $G$ for which $v_1,\dots,v_{H+1}\in\{a,b\}$.

To describe a $k$-face of $\LD(G,\SimpComp)$, we may also use the ordered pair $(F, f)$, where $F$ is a $k$-face of $\SimpComp$, and $f$ is a function mapping each element of $F$ to a vertex of $G$. Because of the local densifier's tensor structure, $\im(f)$ is either a single vertex, or a pair of vertices that form an edge in $G$. 

Linear algebraically, we can think of this graph construction as adding a self loop to each vertex of $G$ and then taking the tensor product with the $1$-skeleton of $\SimpComp$.

Our construction is $\finalHDX := \LD(G,\baseHDX)$, where $\baseHDX$ is equal to $\calK_s^{(H)}$, the $H$-dimensional complete complex on some constant $s \ge H + 1$ vertices, and $G$ is a $T$-regular triangle-free expander graph on $n$ vertices.  We endow $\finalHDX$ with a balanced weight function $w$ induced by  setting the weights of all $H$-faces to $1$.

As a first step to understanding this construction, we inspect the weights induced on $k$-faces for $k < H$.  Consider a $k$-face $F\coloneqq\{(v_1,b_1),\dots,(v_{k+1},b_{k+1})\}$.  A short calculation reveals that if $v_1,\dots,v_{k+1}$ are all equal, then $w(F)$ is equal to $w_{J,k}\coloneqq \binom{s}{H - k} \cdot [T2^{H-k}-(T-1)]$ and otherwise, $w(F)$ is equal to $w_{I,k}\coloneqq \binom{s}{H - k} \cdot [2^{H - k}]$.  Henceforth,  write $w_J$ and $w_I$ instead of $w_{J,k}$ and $w_{I,k}$ when $k$ is understood from context.

We now list out what we prove about $\finalHDX$.  Most importantly, we show:
\begin{theorem}\label{thm:mixing-main}
    For every $1\le k < H$, the Markov transition matrix $\finalHDX^{\down\up}_k$ for down-up (and equivalently up-down) random walks on the $k$-faces satisfies:
    \[
        \Expansion\left(\finalHDX^{\down\up}_k\right) \ge \frac{\Expansion(G)}{64 T^2(k + 1)^2(s - k)(2^k - 1)} \enspace.
    \]
\end{theorem}

We dedicate \pref{sec:spec-gap-HO} to proving \pref{thm:mixing-main}.


As an immediate corollary of \pref{thm:mixing-main} and \pref{thm:spec-gap-to-mixing}, we get that
\begin{corollary}
Let $N_k$ denote the number of $k$-faces in $\finalHDX$. Then the  $\epsilon$-mixing time of $\finalHDX^{\down\up}_k$ satisfies:
\begin{equation*}
    t(\eps) \leq \frac{64 T^2(k + 1)^2(s - k)(2^k - 1)}{\Expansion(G)} \cdot \log \left( \frac{2N_k}{\eps}\right) \enspace.
\end{equation*}
We note that $N_k = \Theta(n)$.
\end{corollary}

We also derive bounds on the expansion of links of $\finalHDX$.  In particular, as a direct consequence of \pref{thm:main} and the discussion of the expansion properties of the complete complex in \pref{eg:complete-complex}, we conclude:
\begin{theorem}\label{thm:local-exp-main}  We can prove the following bounds on the local and global expansion of $\finalHDX$:
    \begin{align*}
        \GlExp(\finalHDX) &\ge \left[\frac{1}{2} - \frac{1}{2 \cdot (T2^H + 1)} \right] \cdot \Expansion(G),\text{ and}\\
        \LocExp(\finalHDX) &\ge \frac{1}{2}.
    \end{align*}
\end{theorem}

\begin{remark}
    Suppose $\bG$ is a random $T$-regular (triangle-free) graph and $H\ge T$.  Then the corresponding (random) simplicial complex $\bcalQ$, as a consequence of Friedman's Theorem \cite{friedman}\footnote{Friedman's theorem says that a random $T$-regular graph, whp, has two-sided spectral gap $\frac{T-2\sqrt{T-1}-o_n(1)}{T}$.  Additionally, random graphs are triangle-free with constant probability.}, with high probability satisfies
    \begin{align*}
        \Expansion\left(\bcalQ_k^{\down\up}\right) &\ge \frac{T-2\sqrt{T-1}-o_n(1)}{64T^3(k+1)^2(s-k)(2^k-1)}\\
        \GlExp(\SimpComp_{\bcalQ}) &\ge \frac{T-2\sqrt{T-1}-o_n(1)}{T+1},\text{ and}\\
        \LocExp(\SimpComp) &\ge 1/2.
    \end{align*}
    Thus, $\bcalQ$ endows a natural distribution over simplicial complexes that gives a high-dimensional expander with high probability.
\end{remark}

\begin{remark}
    If $G$ is \emph{strongly explicit}, such as an expander from \cite{RVW02,BT11}, then $\calQ$ is also strongly explicit since the tensor product of two strongly explicit graphs is also strongly explicit.
\end{remark}

\section{Local Expansion}   \label{sec:local-exp}

For this entire section, we will mainly work with the complex $\LD(G, \SimpComp)$, so when we use $\Link(\cdot)$ without a subscript, it will be with respect to $\LD(G, \SimpComp)$. Next, fix a face $\sigma = (F, f) \in \kface(\LD(G, \SimpComp))$. In order to study the expansion of the $1$-skeleton of $\Link(\sigma)$, we need to first compute the weights on its 1-faces. 

Let $\tau = \{(v_1, b_1), (v_2, b_2)\} \in \Face{2}(\Link(\sigma))$, where as before, $v_i \in V(G)$ and $b_i \in \Face{1}(\SimpComp)$. There are several cases we need to consider:
\begin{enumerate}
    \item Case 1: $|\im(f)| = 2$. \\
    Here, $w_\sigma(\tau) = w(\tau \cup \sigma)$, which is proportional to the number of $H$-faces $(F', f')$ that contain $\tau \cup \sigma$. The face $\tau \cup \sigma$ already has $(k + 3)$ vertices, so there are $\binom{S}{H - (k + 2)}$ possibilities of $F'$. There are $2^{H - (k + 2)}$ choices for $f'$, since $\im(f')$ must equal $\im(f)$.

    \item Case 2: $|\im(f)| = 1$.   
    \begin{enumerate}
    \item Case 2(a): $v_1 = v_2 \in \im(f)$ and $\{b_1, b_2\} \in \Face{2}(\Link_\SimpComp(F))$. \\
    Again, there are $\binom{S}{H - (k + 2)}$ possibilities for $F'$. Since $v_1 = v_2 \in \im(f)$, we will have $T \cdot [2^{H - (k + 2)} - 1] + 1$ choices for $f'$, as $v_1$ has $T$ neighbors in $G$, and when $f'$ is not constant on $v_1$, there are $T$ choices for the other value it can take. 
    
    \item Case 2(b): $v_1 \neq v_2$ but $(v_1, v_2) \in E(G)$, and $\{b_1, b_2\} \in \Face{2}(\Link_\SimpComp(F))$. \\
    Again, we have $\binom{S}{H - (k + 2)}$ possibilities for $F$, but we only have $2^{H - (k + 2)}$ choices for $f'$; the image of $f'$ must be $\{v_1, v_2\}$. 
    
    \item Case 2(c): $v_1 = v_2 \notin \im(f)$ but $v_1 \cup \im(f) \in E(G)$, and $\{b_1, b_2\} \in \Face{2}(\Link_\SimpComp(F))$. \\
    The analysis is identical to that of Case 2(b)
    \end{enumerate}
\end{enumerate}
For simplicity, we'll assign weights to the elements of $\Face{2}(\LD(G, \SimpComp))$ as below:
\begin{align*}
w(\{(v_1, b_1), (v_2, b_2)\}) &=
\begin{cases}
w_{S, k} := 2^{H - (k + 2)} & \text{ for Case 1, 2(b), and 2(c)} \\
w_{C, k} := 1 + T(2^{H - (k + 2)}-1) & \text{ for Case 2(a)} 
\end{cases}
\end{align*}
(Here, the $C$ and $S$ denote ``center'' and ``satellite,'' whose meanings will be more natural when discussing $\Link(\sigma)$ when $\sigma \neq \emptyset$.)

\begin{remark}
Note that if we choose $\sigma = \emptyset$ (so $k = -1$), we simply get the weights of the $1$-skeleton of $\LD(G,\SimpComp)$ itself, which will be useful for computing global expansion. 
\end{remark}

\begin{theorem} \label{thm:main}
    Let $G$ be a triangle-free $T$-regular graph and let $\SimpComp$ be a pure $H$-dimensional simplicial complex.  Then
    \begin{align*}
        \GlExp(\LD(G,\SimpComp)) &= \min\left\{ \frac{\dT 2^{\dH-1}}{\dT2^{\dH}-(\dT-1)} \cdot \Expansion(G),\GlExp(\SimpComp)\right\}, \text{ and}\\
        \LocExp(\LD(G,\SimpComp)) &= \Expansion(\SimpComp).
    \end{align*}
\end{theorem}
\begin{proof}
    Let $\wt{G}$ be the graph obtained by adding self-loops to $G$, with transitions
    \begin{align*}
    \wt{G}[i \to j] &=
        \begin{cases}
        \frac{w_{C, -1}}{w_{C, -1} + Tw_{S, -1}} & \text{ if $i = j$} \\
        \frac{w_{S, -1}}{w_{C, -1} + Tw_{S, -1}} & \text{ otherwise}
        \end{cases}
    \end{align*}
    For large $H$, the self loop probabilities approach $\frac{1}{2}$, while the others approach $\frac{1}{2T}$.
    
    First, observe that $\Adj{\LD(G,\SimpComp)} = \Adj{\wt{G}}\tensor\Adj{\SimpComp}$.  Thus,
    \[
        \Spec(\LD(G,\SimpComp)) = \{\lambda\mu:\lambda\in\Spec(\wt{G}),\mu\in\Spec(\SimpComp)\}.
    \]
    and hence the second largest absolute eigenvalue is no more than $\max\{\lambda_1(\wt{G})|\lambda|_2(\SimpComp), \lambda_1(\SimpComp)|\lambda|_2(\wt{G})\}$, which is simply equal to $\max\{|\lambda|_2(\wt{G}),|\lambda|_2(\SimpComp)\}$.  This implies that
    \begin{align*}
        \GlExp(\LD(G,\SimpComp)) = \min\{\Expansion(\wt{G}),\GlExp(\SimpComp)\}.
    \end{align*}
    By \pref{lem:scaling}, 
    \begin{align*}
    \Expansion(\wt{G}) &= (1 - \frac{w_{C, -1}}{w_{C, -1} + Tw_{S, -1}}) \cdot \Expansion(G) \\
    &= \frac{\dT 2^{\dH-1}}{\dT2^{\dH}-(\dT-1)} \cdot \Expansion(G)
    \end{align*}
    the first part of the theorem statement follows.
    
    \noindent Next, we lower bound $\LocExp(\LD(G,\SimpComp))$.  For any face $S$ in $\LD(G,\SimpComp)$, there exists an edge $\{u,v\}$ in $G$ such that $S$ is contained in $\{u,v\}\times S'$ where $S'$ is a face of $\SimpComp$. If $S$ contains vertices from both $\{u\}\times S'$ and $\{v\}\times S'$, then $\Link(S)$ is isomorphic to $\LD(\mathsf{edge}, \Link(S'))$ where $\mathsf{edge}$ denotes a single-edge graph.
    \[
        \Spec(\LD(\mathsf{edge},\Link(S'))) = \{0\} \cup \Spec(\Link(S'))
    \]
    and hence
    \[
        \Expansion(1\text{-}\mathsf{skeleton}(\Link(S))) = \Expansion(\SimpComp).
    \]
    Without loss of generality, the remaining case is if $S$ contains vertices from only $\{u\}\times S'$.  In this case, $\Link(S)$ is isomorphic to $\LD(\mathsf{star}, \Link(S'))$ where $\mathsf{star}$ denotes a star graph with $T$ satellites.
    \begin{align}   \label{eq:star-tensor}
        \Spec(\LD(\mathsf{star},\Link(S'))) = \{\lambda\mu:\lambda\in\Spec(\Link(S')),\mu\in\Spec(M)\}
    \end{align}
    where $M$ is $\mathsf{star}$ with self loops added on each vertex. We'll call the center vertex of $M$ the ``center'' vertex, and we'll call the remaining vertices the ``satellites.''
    
    Using $w_{C, k}$ and $w_{S, k}$ for Cases 2(a), 2(b), and 2(c) computed above, we can also find the appropriate weights for $M$.
    \begin{align*}
    M[i \to j] &=
        \begin{cases}
        \displaystyle\frac{w_{C, k}}{w_{C, k} + Tw_{S, k}} & \text{ if $i = j$, $i$ is the center}  \\
        \displaystyle\frac{w_{S, k}}{w_{C, k} + Tw_{S, k}} & \text{ if $i$ is the center vertex, $j$ is a satellite} \\
        \displaystyle\frac{1}{2} & \text{ if $i$ is a satellite}
        \end{cases}
    \end{align*}

    We can completely classify the eigenspaces of $\Adj{M}$ and determine their corresponding eigenvalues as follows.
    \begin{enumerate}
        \item The vector with value $\frac{w_{C, k} + Tw_{S, k}}{2 w_{S, k}}$ on the center of the star and $1$ on satellites is an eigenvector of $\Adj{M}$ with eigenvalue $1$.
        
        \item The $(T-1)$-dimensional subspace of vectors which are $0$ on the center of the star, and whose entries sum to $0$ is an eigenspace for eigenvalue $1/2$.
        
        \item The vector with value $-T$ on the center and $1$ on the satellites is an eigenvector with eigenvalue $\frac{1}{2} - \frac{w_{C,k}}{w_{C,k} + Tw_{S, k}}$. For large $H$, this eigenvalue approaches $0$.
    \end{enumerate}
Since the above classification gives $T+1$ eigenvectors it is complete and it is clear that the second largest absolute eigenvalue of $M_2$ is bounded by $1/2$ and thus in this case as well, using \pref{eq:star-tensor}, we can infer
\[
    \Expansion(1\text{-}\mathsf{skeleton}(\Link(S)))\ge\min\{\Expansion(\SimpComp),1/2\}.
\]
which means
\begin{align*}
    \LocExp(\LD(G,\SimpComp)) \ge \min\{\Expansion(\SimpComp),1/2\}.
\end{align*}
\end{proof}

\section{Spectral Gap of High Order Walks}\label{sec:spec-gap-HO}

\subsection{Offsets and Colors}\label{sec: intro to duwalk}


We now inspect the structure of the $k$-faces of our construction $\finalHDX$ in more detail.

\begin{definition}[$k$-faces of $\finalHDX$]
    The set of $k$-faces of $\finalHDX$ is exactly equal to the set of tuples $(F,f)$ where $F$ is a $k$-face of $\baseHDX$ and $f$ is a labeling of each element by endpoints of some edge $\{u,v\}$ in $\baseExp$.  We call $(F, f)$ \emph{$t$-offset} if either $\left|\{x\in F:f(x)=u\}\right| = t$ or $\left|\{x\in F:f(x)=v\}\right|=t$.
\end{definition}

\begin{remark}
    Suppose $t \le k+1-t$.  Note that a $(k+1-t)$-offset state is also $t$-offset, but we will stick to the convention of describing such states as $t$-offset.  For example, a $(k+1)$-offset state is also $0$-offset, but we will only use the term $0$-offset.
\end{remark}

\begin{definition}[Coloring of $k$-faces of $\finalHDX$]
    We color a $k$-face $(F,f)$ of $\finalHDX$ with $\im(f)$.  Each $0$-offset face is then colored with a vertex of $G$ and the remaining faces are each colored with an edge of $G$.
\end{definition}

\myfig{0.8}{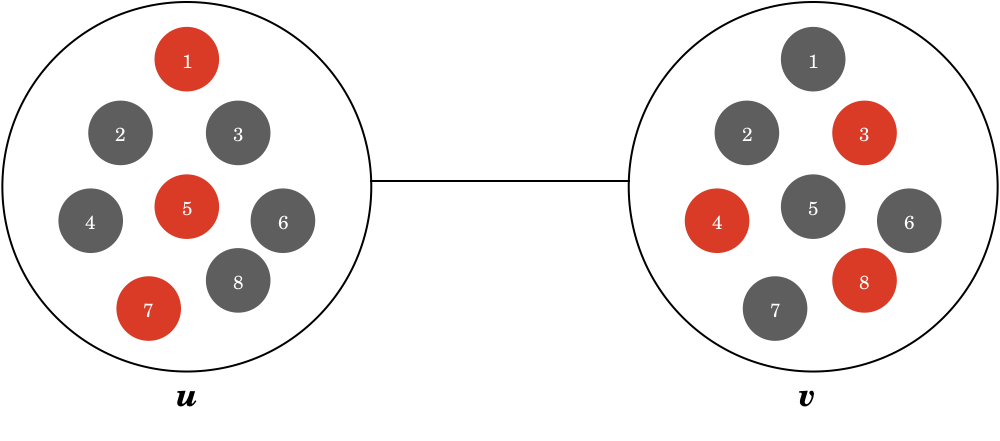}{A 5-face in $\finalHDX$.  Corresponding $5$-face in $\baseHDX$ is $\{1,3,4,5,7,8\}$ is given by red vertices.  Labeling is $(1,u),(3,v),(4,v),(5,u),(7,u),(8,v)$.  $\{u,v\}$ is an edge in $\baseExp$.  Color of $5$-face is $\{u,v\}$}{fig:5-face}

\myfig{0.3}{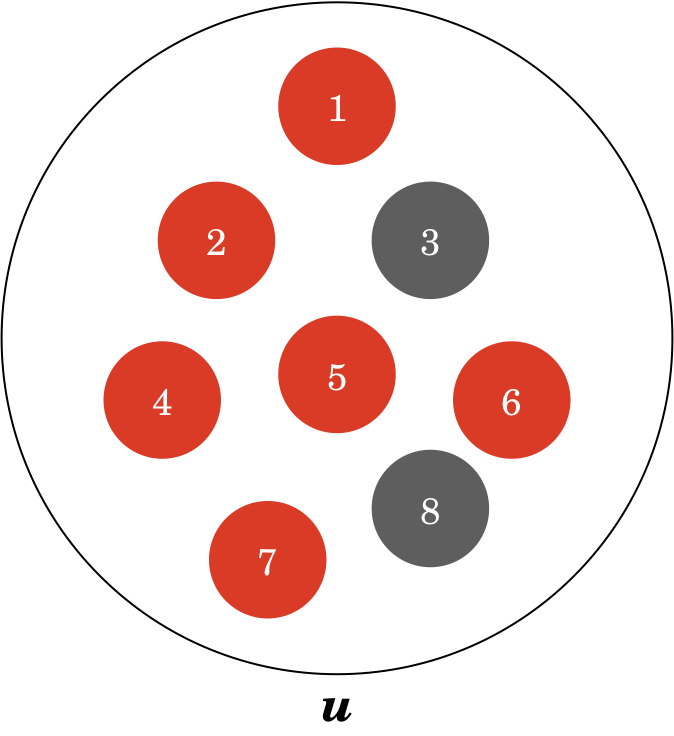}{A $0$-offset $5$-face.  Color of $5$-face is $\{u\}$.}{fig:monochrome-5-face}

In the rest of the section, we study the spectral gap of the Markov chain $\duwalk$, the down-up random walk on $k$-faces of $\finalHDX$ induced by certain special weight functions --- weight functions $w:\kface(\finalHDX)\rightarrow\R_{\ge 0}$ with the property that there are two values $w_I$ and $w_J$ such that
\[
    w((F,f)) =
    \begin{cases}
        w_J &\text{if $(F,f)$ is $0$-offset}\\
        w_I &\text{otherwise}.
    \end{cases}
\]
For instance, if we impose uniform weights on the highest dimensional faces of our complex, the propagated weights on the $k$-th level will satisfy the above property. The $w_I$ and $w_J$ values for this setup is in Appendix \ref{tab:trans-probs-calcs}.

For the sequel, we use $D$ to refer to the quantity $Tw_I+w_J$.  The transition probabilities between states $(F,f)$ and $(F',f')$ depends on a number of conditions such as whether they are $0$-offset or $1$-offset or a different type, whether they arise from the same $k$-face in $\baseHDX$, and the colors of $(F,f)$ and $(F',f')$ respectively.  We provide a detailed treatment of the transition probabilities $\duwalk[(F,f)\to (F',f')]$ in \pref{tab:trans-probs-calcs} in \pref{app:prob-calcs}. From the transition probability table we observe that:
\begin{observation}\label{obs:smallest-eigval}
    For all $k$-faces in $\finalHDX^{\down\up}_k$, the self-loop probability is at least $\frac{1}{s - k}\cdot\frac{w_J}{D}$. Therefore, the smallest eigenvalue of $\finalHDX^{\down\up}_k$ is at least $\frac{1}{s - k}\cdot\frac{w_J}{D} - 1$.
\end{observation}

\subsection{High-Level Picture of $\duwalk$}

As noted in the previous subsection, each $k$-face can be described by three parameters: a base face $F \in \kface(\baseHDX)$, a ``color'' set $C$ that is either a single vertex or an edge in $E(G)$, and a function $f: F \to C$. The walk $\duwalk$ is difficult to analyze directly, but by grouping states based on these three parameters, we can decompose the walk into a projection and restriction chain, and analyze it using the tools from \cite{jerrum}.

\myfig{0.65}{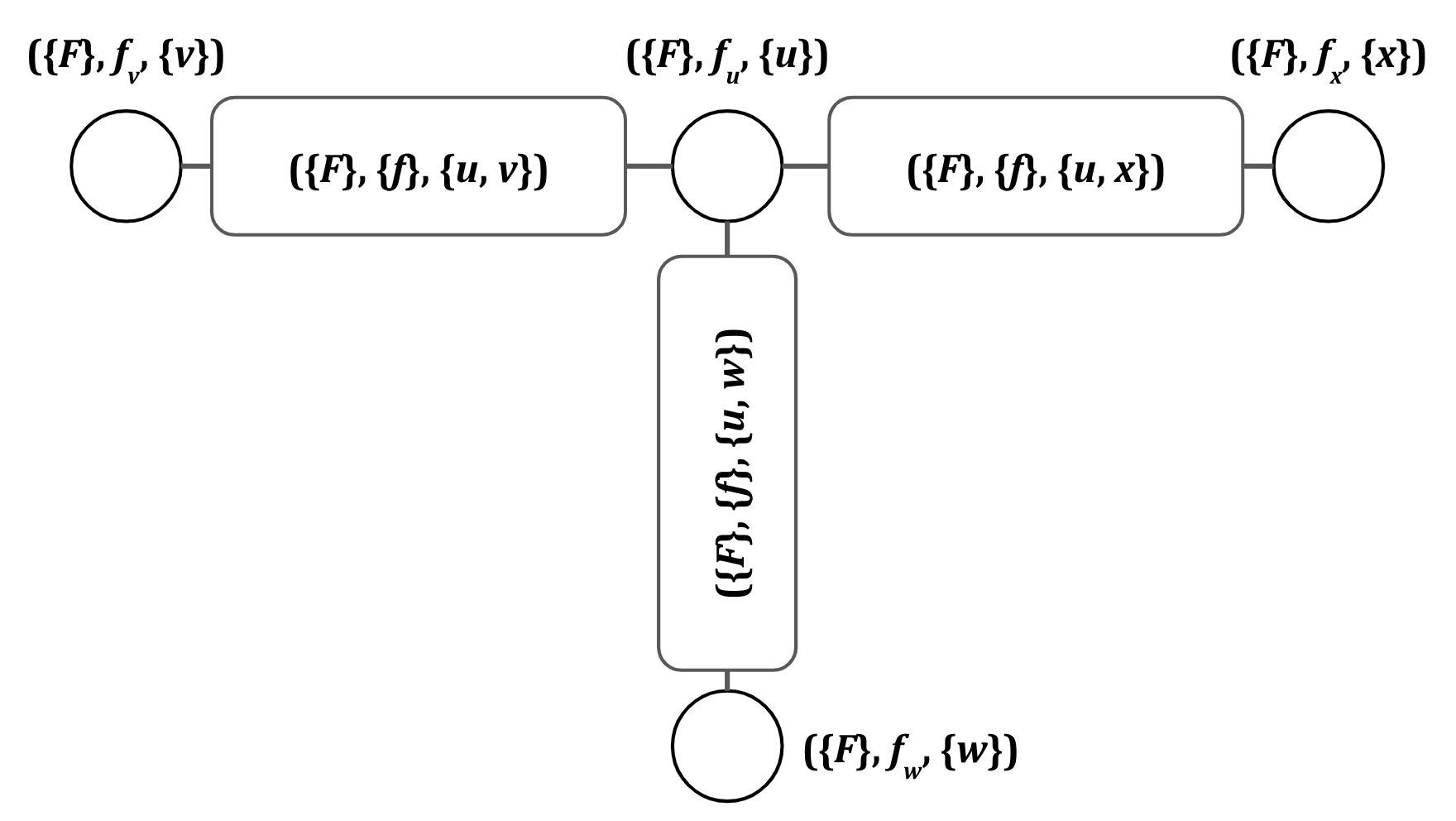}{This figure illustrates $\duwalk$, with states clustered by their color. The rounded rectangles correspond to colors that are edges, while circles correspond to colors that are single vertices. In each cluster, the $\{F\}$ indicates that all $F$ could be represented. Similarly, $\{f\}$ indicates that any $f$ with $\im(f)$ as the color set can be represented. We use $f_u$ to denote the constant function on $u$.}{fig:unsplit} 

At the outermost level, we can first group states into subchains based on their color. All subchains whose color is an edge (the rounded rectangles in \pref{fig:unsplit}) are isomorphic to each other; similarly, all subchains whose color is a single vertex (the circles in \pref{fig:unsplit}) are also isomorphic to each other. At first, it seems promising to partition $\duwalk$ into these subchains; however, it is inconvenient that these subchains are not \emph{all} isomorphic. To remedy this, we split the single-vertex-colored subchains into $\dT$ isomorphic copies (with some changes to transition probabilities), and absorb them into the edge-colored subchains. This is detailed in the next section. 

If we use this partition, the projection chain resembles a random walk on the \emph{line graph} of $G$. Each restriction chain corresponds to all states of a single color $C$. The states are still represented by any base face $F \in \kface(\baseHDX)$ and any function $f: F \to C$. To analyze each of these restriction chains, it is simplest to apply \cite{jerrum} once more. 

Now, we first group states by which base face $F$ they correspond to. The subchains derived from fixing a particular $F$ (the rectangles in \pref{fig:inner-rest}) are all isomorphic to each other, which leads to a much simplified analysis. Using this partition, the projection chain is simply the $k$-down-up walk on $\baseHDX$. Each restriction chain is thus over states corresponding to a fixed base face $F$ and fixed color $C$, but the function $f: F \to C$ is allowed to vary. At this point, we may assume $|C| = 2$; thus $f$ corresponds to assigning every element of $F$ one of two elements. The inner restriction chain can be modeled by a hypercube. 

\myfig{0.5}{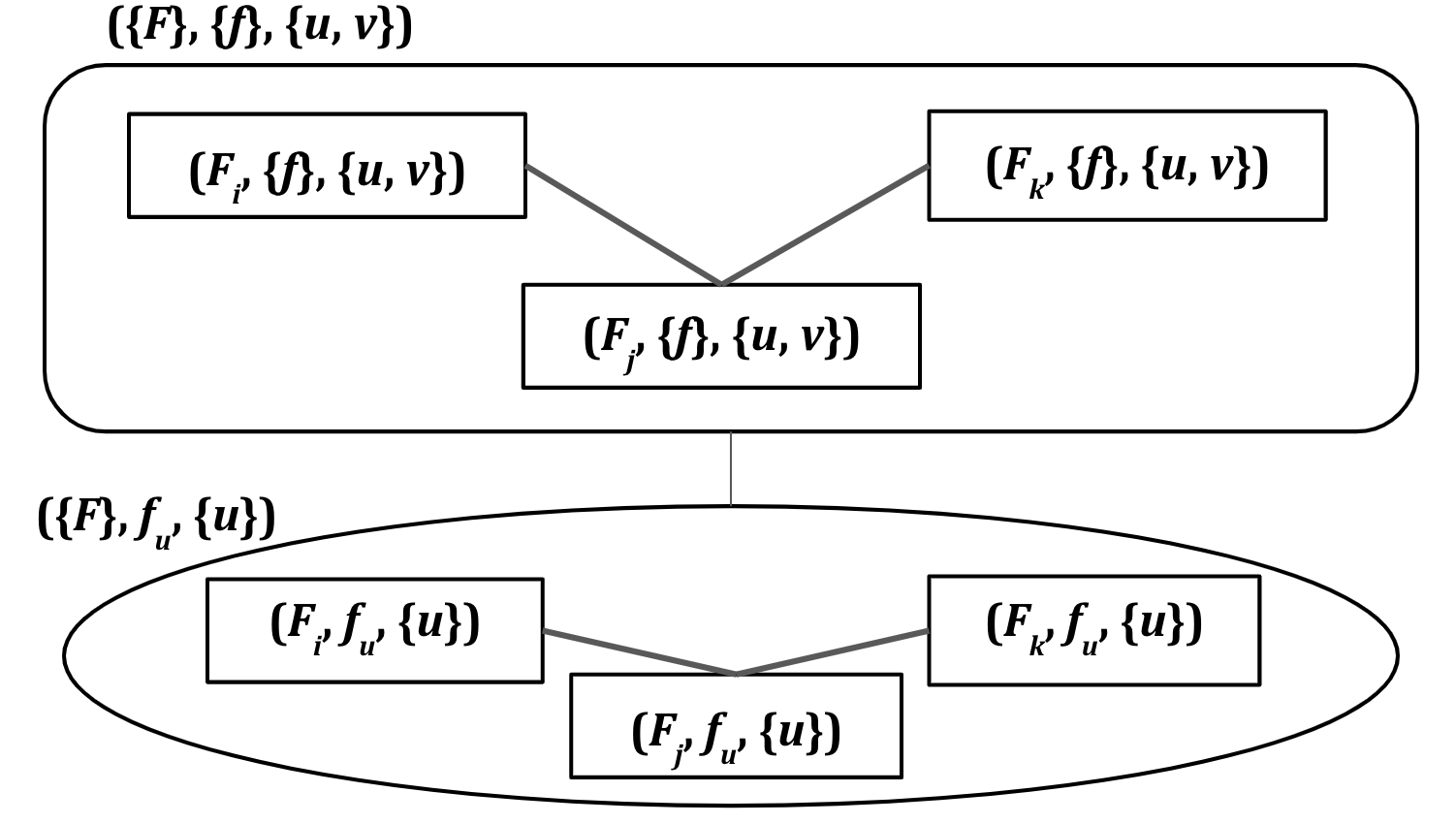}{This figure illustrates a subchain of $\duwalk$, for particular color $\{u, v\}$ and $\{u\}$. We can further cluster the states in this subchain by which face $F$ in $\baseHDX$ they represent. Again, $\{f\}$ indicates that $f$ can be any function with $\im(f)$ as the color.}{fig:inner-rest}

Thus, the spectral gap of $\duwalk$ is a combination of the spectral gaps of (1) the line graph of $G$, (2) the $k$-down-up walk on $\baseHDX$, and (3) the random walk on a hypercube. 

\subsection{Splitting $0$-Offset Vertices}
Towards our end goal of lower bounding the spectral gap of $\duwalk$, we find it convenient to analyze a related Markov chain $\duwalksp$, since the related chain has a natural partition into isomorphic subchains.  $\duwalksp$ has the property that its spectrum contains that of $\duwalk$, which lets us translate a lower bound on the spectral gap of $\duwalksp$ to a lower bound on the spectral gap of $\duwalk$.

\begin{definition}[Split chain $\duwalksp$ and coloring of states in $\duwalksp$] \label{def:split}
    We identify each state in $\States(\duwalksp)$ with a tuple $(F,f,c)$ where $(F,f)$ is a face in $\kface(\finalHDX)$ and $c$ is a color.
    \begin{enumerate}
        \item For each $0$-offset face $(F,f)$ in $\kface(\finalHDX)$, let $\{u\}$ be the color of $F$, and let the neighbors of $u$ in $\baseExp$ be $v_1,\dots,v_T$. $\States(\duwalksp)$ contains the states $(F,f,\{u, v_1\}),\dots,(F,f,\{u, v_T\})$ in place of the state $(F, f, u)$.
        \item For each remaining $k$-face $(F,f)$ of $\finalHDX$ (i.e. each $k$-face that isn't $0$-offset), $\States(\duwalksp)$ contains $(F,f,\im(f))$.
    \end{enumerate}
    For each pair of states $(F,f,c),(F',f',c')$ in $\States(\duwalksp)$,
    \[
        \duwalksp[(F,f,c)\to(F',f',c')] =
        \begin{cases}
            \frac{\duwalk[(F,f)\to (F',f')]}{T} &\text{if $(F',f')$ is $0$-offset}\\
            \duwalk[(F,f)\to (F',f')] &\text{otherwise.}
        \end{cases}
    \]
\end{definition}

Intuitively, we want to split any transition to a $0$-offset face in $\finalHDX$ into $\dT$ separate transitions in $\duwalksp$, since each $0$-offset face is also split into $\dT$ new states.

\myfig{0.8}{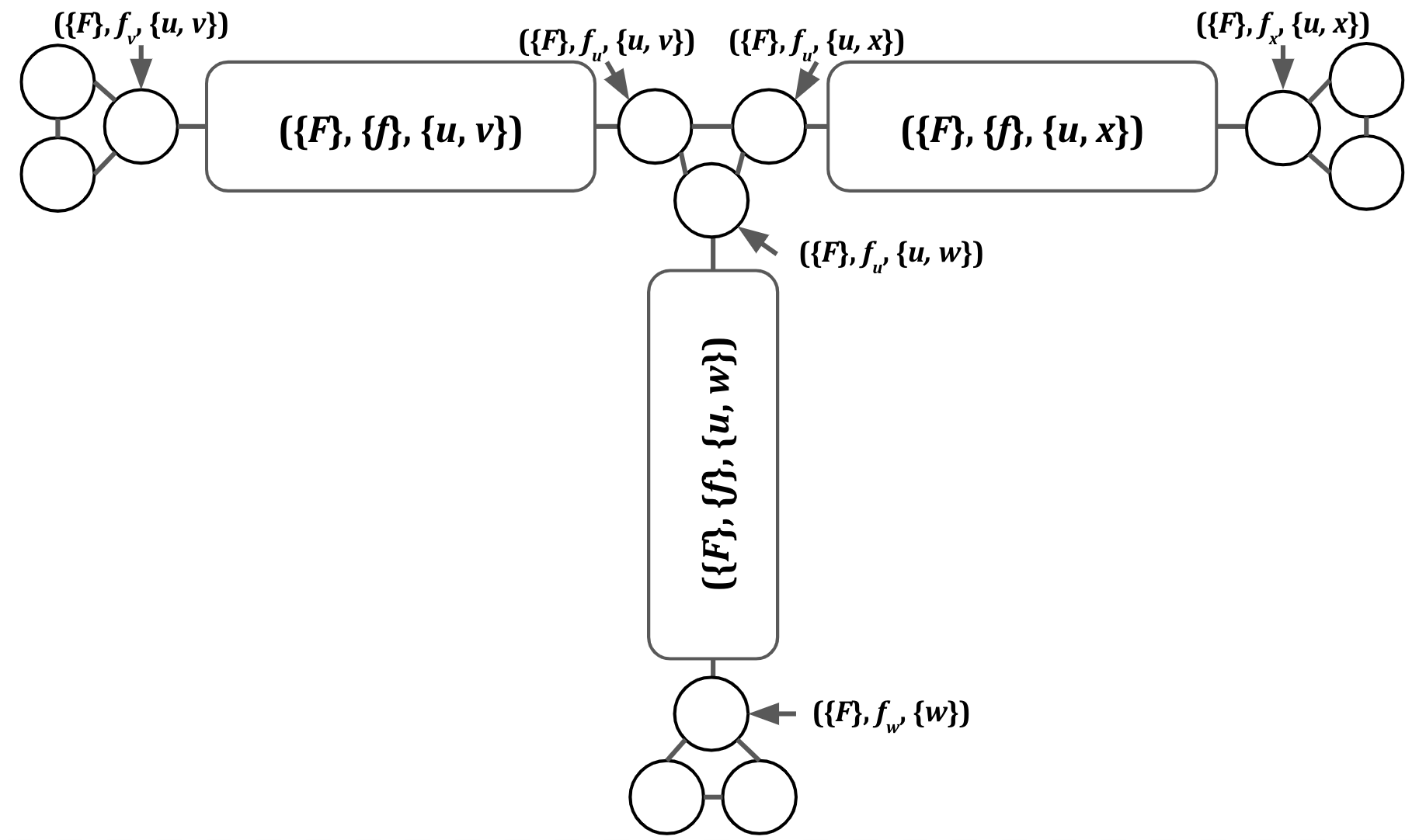}{This figure illustrates the post-split vertices of \pref{def:split}. The new vertices can take on any $F$, but their mappings $f$ will be constant functions.}{fig:split}

\begin{definition}
    We say two $k$-faces $(F,f,e)$ and $(F',f',e')$ have \emph{identical base $k$-faces} if $F=F'$ and \emph{different base $k$-faces} if $F\ne F'$.
\end{definition}

\begin{definition}  \label{def:lonely}
    Given a state $(F,f,e)$ such that $(F,f)$ is a $1$-offset face, there is a single vertex $v$ such that $f(v)$ is different from $f(u)$ for all $u$ in $F\setminus\{v\}$.  We call this vertex $v$ a \emph{lonely} vertex.
\end{definition}

In the next lemma, we show that the spectrum of the original Markov chain $\duwalk$ is contained in that of $\duwalksp$. 
\begin{lemma}   \label{lem:split-reduction}
    $\spec\left(\duwalk\right)\subseteq\spec\left(\duwalksp\right)$, and therefore, $\lambda_2(\duwalk)\le\lambda_2(\duwalksp)$.
\end{lemma}
The proof can be found in \pref{app:proof-of-split-reduction}.

\subsubsection{Stationary Distribution of $\duwalksp$} \label{stationary-main-section}

If we want to apply the projection and restriction framework to $\duwalksp$, we first need to compute its stationary distribution. To do this, we take advantage of the time-reversibility of the high order random walks, and apply the detailed balance equations. The transition probabilities in $\duwalksp$ are laid out in detail in \pref{app:prob-calcs}.

\begin{lemma} \label{lem:stationary-main-lemma}
    The stationary distribution of the Markov chain $\duwalksp$ is given by:
    \begin{align*}
        \pi_{\duwalksp}(x) &=
        \begin{cases} 
        \displaystyle\frac{1}{|E(G)|} \cdot \frac{1}{\binom{s}{k + 1}} \cdot \frac{1}{2} \cdot \frac{w_J}{(2^k - 1)Tw_I + w_J} & \text{ for $x$ $0$-offset} \\
        \displaystyle\frac{1}{|E(G)|} \cdot \frac{1}{\binom{s}{k + 1}} \cdot \frac{1}{2} \cdot \frac{Tw_I}{(2^k - 1)Tw_I + w_J} & \text{ otherwise} 
        \end{cases}
    \end{align*}
\end{lemma}
\begin{proof}
Via the detailed balance equations, we first observe that all vertices with the same offset have the same stationary distribution. Now, let $x$ be any $0$-offset vertex and $y$ be any $1$-offset vertex. Using the detailed balance equations, we have:
$$
\pi_{\duwalksp}(x) \cdot \frac{w_I}{(k + 1)(s - k)D} = \pi_{\duwalksp}(y) \cdot \frac{w_J}{(k + 1)(s - k)DT}
$$
Now, let $x$ be any $t$-offset vertex, with $t \geq 1$, and let $y$ be any $(t + 1)$-offset vertex. Again, using the detailed balance equations:
\begin{align*}
    \pi_{\duwalksp}(x) \cdot \frac{1}{2(k + 1)(s - k)} &= \pi_{\duwalksp}(y) \cdot \frac{1}{2(k + 1)(s - k)}
\end{align*}
From here, we see that all $0$-offset faces have one stationary distribution probability, and all other faces also share the same stationary probability. The relations above tell us that for a $0$-offset vertex $x$, and a $t$-offset vertex $y$ with $t \geq 1$:
$$
\frac{\pi_{\duwalksp}(x)}{\pi_{\duwalksp}(y)} = \frac{w_J}{Tw_I}
$$
Normalizing so that $\sum_{x \in \duwalksp} \pi_{\duwalksp}(x) = 1$ gives the desired result.
\end{proof}

\subsection{Outer Projection and Restriction Chains}

Now, we can further decompose $\duwalksp$ into a projection chain and $m$ isomorphic restriction chains, where $m = \abs{E(\baseExp)}$, since we will have one partition element for each edge in $G$. Formally, we partition $\States(\duwalksp)$ into $m$ disjoint sets $\Omega_{1} \cup \dots \cup \Omega_{m}$, where $\Omega_i = \{(F, f, c) \mid c = e_i\}$. 

\subsubsection{The Outer Projection Chain}
The partition $\Omega$ induces a projection chain $([m], \oproj)$. The state space is $[m]$. The edge set is 
$$
    E(\oproj) = \{\{i,j\}\mid \exists (F,f,e_i)\in \Omega_i \text{ and } (G,g,e_j)\in\Omega_j \text{ s.t. } \duwalksp[(F,f,e_i)\to (G,g,e_j)] > 0\}
$$
In words, we have an edge between $i$ and $j$ if there are transitions from $\Omega_i$ to $\Omega_j$.

\noindent We obtain the following lower bound on the spectral gap of $\oproj$. 
\begin{lemma}\label{lem:oproj-spectral-gap}
    The spectral gap of $\oproj$ is $$\frac{\Expansion(\baseExp)}{2}\cdot \frac{w_J+\dT w_I}{w_J+(2^{k}-1)\dT w_I} \geq \frac{\Expansion(\baseExp)}{2(2^k - 1)}\enspace.$$
\end{lemma}
A detailed account of the transitional probabilities of the projection chain $([m], \oproj)$ can be found in \pref{app:oproj-tp} and the proof of the lemma can be found in \pref{app:proof-of-oproj-spectral-gap}.

\subsubsection{The Outer Restriction Chain}

Each partition block $\Omega_i$ induces a restriction chain $\orest{i}$. We show that all restriction chains $\orest{i}$ for $i\in[m]$ are isomorphic.
\begin{lemma} \label{lem:orest-iso}
    For any $i\neq j$, $i,j\in [m]$, the restriction chains $\orest{i}$ and $\orest{j}$ are isomorphic.
\end{lemma}
\noindent The proof is in \pref{app:proof-of-orest-iso}. The transition probabilities of a restriction chain $\orest{i}$ is deduced in \pref{app:orest-tp}.

\subsubsection{Stationary Distribution of $\orest{1}$}

To compute the spectral gap of $\orest{1}$, we will further decompose the chain in the next section. In order to apply the projection and restriction framework once more to $\orest{1}$, we must again compute a stationary distribution. 
\begin{lemma} \label{lem:stationary-orest}
    The stationary distribution of the outer restriction chain is given by:
    \begin{align*}
        \pi_{\orest{1}}(x) &=
        \begin{cases} 
        \displaystyle\frac{1}{\binom{s}{k + 1}} \cdot \frac{1}{2} \cdot \frac{w_J}{(2^k - 1)Tw_I + w_J} & \text{ for $x$ 0-offset} \\
        \displaystyle\frac{1}{\binom{s}{k + 1}} \cdot \frac{1}{2} \cdot \frac{Tw_I}{(2^k - 1)Tw_I + w_J} & \text{ otherwise} 
        \end{cases}
    \end{align*}
\end{lemma}
\begin{proof}
    By \pref{fact:inherit-time-reversibility}, $\orest{1}$ is also time-reversible. We proceed using the same analysis we used for \pref{lem:stationary-main-lemma}. At the very end, we use a slightly different normalization to get the desired result.
\end{proof}

\subsection{Inner Projection and Restriction Chains}

Now, we are left to study the outer restriction chain, which, for a fixed $e \in E(G)$, is composed of all $(F, f, e)$ in $\States(\duwalksp)$.  Again, we further decompose this chain into projection and restriction chains which are easier to analyze. 

We group all $(F, f, e)$ with the same $F \in \Face{k}(\baseHDX)$ into the same restriction state space $\Omega_F$, which induces a projection chain resembling $\duwalkbase$, the down-up walk on $k$-faces of $\baseHDX$, and a restriction chain resembling a lazy random walk on a $(k+1)$-dimensional hypercube.

\subsubsection{The Projection Chain}
By defining the projection restriction chains as above, we end up with isomorphic restriction chains for each $F \in \Face{k}(\baseHDX)$.  Thus, we can identify each of the states of the inner projection chain $\iproj$ with some face $F \in \Face{k}(\SimpComp)$. Let $\{F_i\}$ be this partition based on face.

Given $F, F' \in \Face{k}(\SimpComp)$, we can only transition from $F$ to $F'$ either when $F = F'$, or when $F \cap F' \in \Face{(k - 1)}$.  This coincides with the feasible transitions in $\duwalkbase$.  

We detail the transition probabilities in $\duwalkbase$ in \pref{app:inner-proj-prob} and are able to obtain the following bounds on the spectral gap of the outer projection chain:
\begin{lemma}\label{lem:iproj-spectral-gap}
$\displaystyle\gap(\iproj) \ge \frac{1}{2T(k + 1)} \enspace.$
\end{lemma}
The proof of \pref{lem:iproj-spectral-gap} can be found in \pref{app:inner-proj-gap}.

\subsubsection{The Restriction Chain}

Each restriction chain $\irest$ can be treated as a $(k + 1)$-dimensional hypercube with self loops. To see this, note that each restriction chain is a set of states $(F,f,e)$ in $\duwalksp$ where both $F$ and $e$ are the same. There are thus $2^{k + 1}$ states in each restriction chain, since for each $x$, we have two choices for $f(x)$.  Associating $x$ where $f(x) = u$ to a $0$-coordinate in a hypercube vertex, and $x$ where $f(x) = v$ to a $1$-coordinate, gives us a bijection from the restriction chain to the hypercube. 

The transition probabilities are summarized in \pref{app:irest-prob}, and we show:
\begin{lemma}\label{lem:irest-spectral-gap}
    If we impose uniform weights on the highest order faces,
    \[
        \gap(\irest) \geq \frac{w_J}{2 \dT w_I} \cdot \frac{2w_J}{D(k + 1)(s - k)} \geq \frac{1}{(k + 1)(s - k)} \enspace.
    \]
\end{lemma}
We defer the proof of \pref{lem:irest-spectral-gap} to \pref{app:irest-gap}.  We also give relevant background in \pref{app:var-char-gap}.

\subsection{Rapid Mixing for High Order Random Walks}

Now we put together the decomposition theorem, the lower bounds for the spectral gaps of the project and restriction chains, and \pref{obs:smallest-eigval} to obtain the following lower bound on the two-sided spectral gap of $\duwalksp$:

\begin{theorem}[Restatement of \pref{thm:mixing-main}]\label{thm:down-up-random-walk-spectral-gap}
    The $k$ down-up random walk $\duwalk$ has one-sided spectral gap,
    \begin{align}   \label{ineq:two-sided}
        \Expansion(\duwalk) \geq \frac{\Expansion(G)}{64 \dT(k + 1)^2(s - k)(2^k - 1)} \enspace.
    \end{align}
\end{theorem}

\begin{proof}
Use $\gap(M)$ to denote the spectral gap of a Markov chain $M$. We deduce from \pref{lem:split-reduction} and \pref{thm:Jerrum-et-al} that
\begin{align*}
    \gap(\duwalk) &\geq \gap(\duwalksp)  &&(\text{\pref{lem:split-reduction}}) \\
                       &\geq \min\left\{\frac{\gap(\oproj)}{3}, \frac{\gap(\oproj)\gap(\orest{1})}{3\gamma_o + \gap(\oproj)}\right\}  &&(\text{\pref{thm:Jerrum-et-al} on }\duwalksp ) \\
                       &\geq \min 
                       \bigg\{\frac{\gap(\oproj)}{3},\\ &\frac{\gap(\oproj)}{3\gamma_o + \gap(\oproj)}\cdot\frac{\gap(\iproj)}{3},\\ &\frac{\gap(\oproj)}{3\gamma_o + \gap(\oproj)}\cdot\frac{\gap(\iproj)\gap(\irest)}{3\gamma_I+\gap(\iproj)} 
                       \bigg\} &&(\text{\pref{thm:Jerrum-et-al} on }\orest{1}), \\
\end{align*}

where 
\begin{align*}
    \gamma_o &= \max_{i\in[m]}\max_{x\in\Omega_i}\sum_{y\in\Omega\setminus\Omega_i} \duwalksp(x,y) < 1 \\
    \gamma_I &= \max_{F\in\Face{k}(\SimpComp)} \max_{x\in V(\irest)} \sum_{y\in V(\orest{1}) \setminus V(\irest)} \orest{1}(x,y) < 1
            .
\end{align*}

Furthermore, \pref{lem:oproj-spectral-gap}, \pref{lem:iproj-spectral-gap} and \pref{lem:irest-spectral-gap} provide lower bounds for $\gap(\oproj)$, $\gap(\iproj)$, and $\gap(\irest)$. If we substitute the spectral-gap lower bounds, and an upper bound of $1$ for both $\gamma_o$ and $\gamma_I$, we obtain a lower bound on $\gap(\duwalk)$:
\begin{align}   \label{ineq:one-sided-gap}
    \gap(\duwalk) &\geq \frac{\Expansion(G)}{64 \dT(k + 1)^2(s - k)(2^k - 1)}.
\end{align}
\pref{obs:smallest-eigval} gives a lower bound on $1-|\lambda_{\min}(\duwalk)|$ larger than the right hand side of \pref{ineq:one-sided-gap}, which immediately lets us upgrade the statement \pref{ineq:one-sided-gap} to \pref{ineq:two-sided}, thus proving the theorem. 
\end{proof}

\section{Future Work}

Our construction and its analysis opens the door for other combinatorial and randomized candidates for high dimensional expanders. Due to the numerous ways of constructing one-dimensional expander graphs via graph products, such as the replacement product and the zig-zag product, a natural direction to pursue is to see if there are high-dimensional analogues of these products as well. Additionally, we traded off arbitrarily good local spectral expansion in favor of having a randomized, simple combinatorial construction. It would be worthwhile to investigate any potential improvements on the $\frac{1}{2}$ local spectral expansion and to understand whether $\frac{1}{2}$ is a natural barrier for any graph--product--based construction. Lastly, though we were able to demonstrate rapid mixing, our analysis relied heavily on the \cite{jerrum} framework, which may not yield a tight bound on the spectral gap of the higher order walks. It would also be interesting to investigate either simpler analyses, or tighter analyses. 

More importantly, our construction demonstrates a large family of high dimensional expanders whose higher order walks mix rapidly, yet do not have arbitrarily good local spectral expansion. This suggests that a constant local spectral expansion may be enough to recover the rapid mixing. Thus far, the predominant machinery for establishing rapid mixing of higher order walks is through Theorem 5.4 of \cite{kaufman-oppenheim}, which requires the second largest eigenvalue of the links to be $o(1)$. It would be interesting to see whether in the regime of constant local spectral expansion there is a decomposition theorem that establishes rapid mixing of higher order walks.

Another feature of our construction is that only an exponentially small fraction (in link size) of the links actually have edge expansion $\frac{1}{2}$. The vast majority of the links, when considering their underlying $1$-skeletons, are in fact complete graphs, which have excellent expansion. In fact, their second eigenvalues will always be negative; if we ignore the exponentially small number of problematic links, we can actually use Theorem 5.4 of \cite{kaufman-oppenheim}. It would be interesting to further explore whether (1) constant local spectrum suffices, or (2) we actually need a large fraction of the links to have $o(1)$ local spectral expansion. 


\section*{Acknowledgements}

We thank Tom Gur for introducing us to this intriguing question and for helpful discussions, and we thank Nikhil Srivastava for insightful conversations.

We would also like to thank the Simons Institute for the Theory of Computing where a large portion of this work was done. The third author is supported by the National Science Foundation Graduate Research Fellowship under Grant No. DGE 1752814.


\bibliography{main}
\appendix
\section{Transition Probabilities of the Down-Up Walk}   \label{app:prob-calcs}

If we impose uniform weights at the highest order faces, then 
\begin{align*}
w_I &= 2^{H-k} \\ 
w_J &= T2^{H-k} - (T-1)
\end{align*}
For ease of notation, we use define another variable $D = Tw_I + w_J$, which will arise very often. Note that for the uniform weights case, $w_J$ is only slightly smaller than $Tw_I$, which will help with some of our asymptotics.

\begin{table}[h]
\caption{Transition Probabilities in $\duwalk$}\label{tab:trans-probs-calcs}
\begin{tabular}{|l|l|l|l|l|l|l|}
\hline
Source                         & Delete                    & Target                          & Same $k$-face &  Same edge & Probability                                           & Count              \\ \hline
\multirow{4}{*}{$0$-offset}  & \multirow{4}{*}{anything} & \multirow{2}{*}{$0$-offset}   & Yes           & N/A       & $\frac{1}{s-k}\cdot\frac{w_J}{D}$                     & $1$                \\ \cline{4-7} 
                               &                           &                                 & No            & N/A       & $\frac{1}{k+1}\cdot\frac{1}{s-k}\cdot\frac{w_J}{D}$   & $(s-(k+1))(k+1)$   \\ \cline{3-7} 
                               &                           & \multirow{2}{*}{$1$-offset}   & Yes           & N/A       & $\frac{1}{k+1}\cdot\frac{1}{s-k}\cdot\frac{w_I}{D}$   & $(k+1)T$           \\ \cline{4-7} 
                               &                           &                                 & No            & N/A       & $\frac{1}{k+1}\cdot\frac{1}{s-k}\cdot\frac{w_I}{D}$   & $(k+1)T(s-(k+1))$  \\ \hline
\multirow{10}{*}{$1$-offset} & \multirow{6}{*}{minority}   & \multirow{2}{*}{$0$-offset}   & Yes           & N/A       & $\frac{1}{k+1}\cdot\frac{1}{s-k}\cdot\frac{w_J}{D}$   & $1$                \\ \cline{4-7} 
                               &                           &                                 & No            & N/A       & $\frac{1}{k+1}\cdot\frac{1}{s-k}\cdot\frac{w_J}{D}$   & $s-(k+1)$          \\ \cline{3-7} 
                               &                           & \multirow{4}{*}{$1$-offset}   & Yes           & Yes       & $\frac{1}{k+1}\cdot\frac{1}{s-k}\cdot\frac{w_I}{D}$   & $1$                \\ \cline{4-7} 
                               &                           &                                 & No            & Yes       & $\frac{1}{k+1}\cdot\frac{1}{s-k}\cdot\frac{w_I}{D}$   & $s-(k+1)$          \\ \cline{4-7} 
                               &                           &                                 & Yes           & No        & $\frac{1}{k+1}\cdot\frac{1}{s-k}\cdot\frac{w_I}{D}$   & $T-1$              \\ \cline{4-7} 
                               &                           &                                 & No            & No        & $\frac{1}{k+1}\cdot\frac{1}{s-k}\cdot\frac{w_I}{D}$   & $(T-1)(s-(k+1))$   \\ \cline{2-7} 
                               & \multirow{4}{*}{majority}     & \multirow{2}{*}{$1$-offset}   & Yes           & Yes       & $\frac{k}{k+1}\cdot\frac{1}{s-k}\cdot\frac{1}{2}$     & $1$                \\ \cline{4-7} 
                               &                           &                                 & No            & Yes       & $\frac{1}{k+1}\cdot\frac{1}{s-k}\cdot\frac{1}{2}$     & $(s-(k+1))k$       \\ \cline{3-7} 
                               &                           & \multirow{2}{*}{$2$-offset}   & Yes           & Yes       & $\frac{1}{k+1}\cdot\frac{1}{s-k}\cdot\frac{1}{2}$     & $k$                \\ \cline{4-7} 
                               &                           &                                 & No            & Yes       & $\frac{1}{k+1}\cdot\frac{1}{s-k}\cdot\frac{1}{2}$     & $k(s-(k+1))$       \\ \hline
\multirow{8}{*}{$t$-offset}  & \multirow{4}{*}{minority} & \multirow{2}{*}{$t$-offset}   & Yes           & Yes       & $\frac{t}{k+1}\cdot\frac{1}{s-k}\cdot\frac{1}{2}$     & 1                  \\ \cline{4-7} 
                               &                           &                                 & No            & Yes       & $\frac{1}{k+1}\cdot\frac{1}{s-k}\cdot\frac{1}{2}$     & $t(s-(k+1))$       \\ \cline{3-7} 
                               &                           & \multirow{2}{*}{$t-1$-offset} & Yes           & Yes       & $\frac{1}{k+1}\cdot\frac{1}{s-k}\cdot\frac{1}{2}$     & $t$                \\ \cline{4-7} 
                               &                           &                                 & No            & Yes       & $\frac{1}{k+1}\cdot\frac{1}{s-k}\cdot\frac{1}{2}$     & $t(s-(k+1))$       \\ \cline{2-7} 
                               & \multirow{4}{*}{majority}     & \multirow{2}{*}{$t$-offset}   & Yes           & Yes       & $\frac{k+1-t}{k+1}\cdot\frac{1}{s-k}\cdot\frac{1}{2}$ & 1                  \\ \cline{4-7} 
                               &                           &                                 & No            & Yes       & $\frac{1}{k+1}\cdot\frac{1}{s-k}\cdot\frac{1}{2}$     & $(k+1-t)(s-(k+1))$ \\ \cline{3-7} 
                               &                           & \multirow{2}{*}{$t+1$-offset} & Yes           & Yes       & $\frac{1}{k+1}\cdot\frac{1}{s-k}\cdot\frac{1}{2}$     & $k+1-t$            \\ \cline{4-7} 
                               &                           &                                 & No            & Yes       & $\frac{1}{k+1}\cdot\frac{1}{s-k}\cdot\frac{1}{2}$     & $(k+1-t)(s-(k+1))$ \\ \hline
\end{tabular}
\end{table}

\newpage


\begin{table}[h]
\caption{Transition Probabilities in $\duwalksp$}\label{tab:trans-probs-calcs-split}
\begin{tabular}{|l|l|l|l|l|l|l|}
\hline
Source                       & Delete                    & Target                        & Same $k$-face & Same edge & Probability                                                           & Count                 \\ \hline
\multirow{8}{*}{$0$-offset}  & \multirow{8}{*}{anything} & \multirow{2}{*}{$0$-offset}   & \multirow{2}{*}{Yes} & Yes       & \multirow{2}{*}{$\frac{1}{s-k}\cdot\frac{w_J}{DT}$}                   & $1$                   \\ \cline{5-5} \cline{7-7} 
                             &                           &                               &                      & No        &                                                                       & $T-1$                 \\ \cline{3-7} 
                             &                           &  \multirow{2}{*}{$0$-offset}  & \multirow{2}{*}{No}  & Yes       & \multirow{2}{*}{$\frac{1}{k+1}\cdot\frac{1}{s-k}\cdot\frac{w_J}{DT}$} & $(s-(k+1))(k+1)$      \\ \cline{5-5} \cline{7-7} 
                             &                           &                               &                      & No        &                                                                       & $(T-1)(s-(k+1))(k+1)$ \\ \cline{3-7} 
                             &                           & \multirow{2}{*}{$1$-offset}   & \multirow{2}{*}{Yes} & Yes       & \multirow{4}{*}{$\frac{1}{k+1}\cdot\frac{1}{s-k}\cdot\frac{w_I}{D}$}  & $(k+1)$               \\ \cline{5-5} \cline{7-7} 
                             &                           &                               &                      & No        &                                                                       & $(k+1)(T-1)$          \\ \cline{3-5} \cline{7-7} 
                             &                           & \multirow{2}{*}{$1$-offset}   & \multirow{2}{*}{No}  & Yes       &                                                                       & $(k+1)(s-(k+1))$      \\ \cline{5-5} \cline{7-7} 
                             &                           &                               &                      & No        &                                                                       & $(k+1)(T-1)(s-(k+1))$ \\ \hline
\multirow{12}{*}{$1$-offset} & \multirow{8}{*}{minority}   & \multirow{4}{*}{$0$-offset}   & \multirow{2}{*}{Yes} & Yes       & \multirow{4}{*}{$\frac{1}{k+1}\cdot\frac{1}{s-k}\cdot\frac{w_J}{DT}$} & $1$                   \\ \cline{5-5} \cline{7-7} 
                             &                           &                               &                      & No        &                                                                       & $T-1$                 \\ \cline{4-5} \cline{7-7} 
                             &                           &                               & \multirow{2}{*}{No}  & Yes       &                                                                       & $s-(k+1)$             \\ \cline{5-5} \cline{7-7} 
                             &                           &                               &                      & No        &                                                                       & $(s-(k+1))(T-1)$      \\ \cline{3-7} 
                             &                           & \multirow{4}{*}{$1$-offset}   & Yes                  & Yes       & \multirow{4}{*}{$\frac{1}{k+1}\cdot\frac{1}{s-k}\cdot\frac{w_I}{D}$}  & $1$                   \\ \cline{4-5} \cline{7-7} 
                             &                           &                               & No                   & Yes       &                                                                       & $s-(k+1)$             \\ \cline{4-5} \cline{7-7} 
                             &                           &                               & Yes                  & No        &                                                                       & $T-1$                 \\ \cline{4-5} \cline{7-7} 
                             &                           &                               & No                   & No        &                                                                       & $(T-1)(s-(k+1))$      \\ \cline{2-7} 
                             & \multirow{4}{*}{majority}     & \multirow{2}{*}{$1$-offset}   & Yes                  & Yes       & $\frac{k}{k+1}\cdot\frac{1}{s-k}\cdot\frac{1}{2}$                     & $1$                   \\ \cline{4-7} 
                             &                           &                               & No                   & Yes       & \multirow{3}{*}{$\frac{1}{k+1}\cdot\frac{1}{s-k}\cdot\frac{1}{2}$}    & $(s-(k+1))k$          \\ \cline{3-5} \cline{7-7} 
                             &                           & \multirow{2}{*}{$2$-offset}   & Yes                  & Yes       &                                                                       & $k$                   \\ \cline{4-5} \cline{7-7} 
                             &                           &                               & No                   & Yes       &                                                                       & $k(s-(k+1))$          \\ \hline
\multirow{8}{*}{$t$-offset}  & \multirow{4}{*}{minority} & \multirow{2}{*}{$t$-offset}   & Yes                  & Yes       & $\frac{t}{k+1}\cdot\frac{1}{s-k}\cdot\frac{1}{2}$                     & 1                     \\ \cline{4-7} 
                             &                           &                               & No                   & Yes       & \multirow{3}{*}{$\frac{1}{k+1}\cdot\frac{1}{s-k}\cdot\frac{1}{2}$}    & $t(s-(k+1))$          \\ \cline{3-5} \cline{7-7} 
                             &                           & \multirow{2}{*}{$t-1$-offset} & Yes                  & Yes       &                                                                       & $t$                   \\ \cline{4-5} \cline{7-7} 
                             &                           &                               & No                   & Yes       &                                                                       & $t(s-(k+1))$          \\ \cline{2-7} 
                             & \multirow{4}{*}{majority}     & \multirow{2}{*}{$t$-offset}   & Yes                  & Yes       & $\frac{k+1-t}{k+1}\cdot\frac{1}{s-k}\cdot\frac{1}{2}$                 & 1                     \\ \cline{4-7} 
                             &                           &                               & No                   & Yes       & \multirow{3}{*}{$\frac{1}{k+1}\cdot\frac{1}{s-k}\cdot\frac{1}{2}$}    & $(k+1-t)(s-(k+1))$    \\ \cline{3-5} \cline{7-7} 
                             &                           & \multirow{2}{*}{$t+1$-offset} & Yes                  & Yes       &                                                                       & $k+1-t$               \\ \cline{4-5} \cline{7-7} 
                             &                           &                               & No                   & Yes       &                                                                       & $(k+1-t)(s-(k+1))$    \\ \hline
\end{tabular}
\end{table}

\section{Spectrum of $\duwalksp$: Proof of \pref{lem:split-reduction}} 

\label{app:proof-of-split-reduction}
        Given a right eigenvector $v$ of $\duwalk$ for eigenvalue $\lambda$, we exhibit a right eigenvector $\wt{v}$ of $\duwalksp$, also for eigenvalue $\lambda$.  Let
    \[
        \wt{v}[(F,f,c)] =
        \begin{cases}
            \frac{v[(F,f)]}{T} &\text{if $(F,f)$ is $0$-offset}\\
            v[(F,f)] &\text{otherwise.}
        \end{cases}
    \]
    We now verify that $\wt{v}$ is indeed a right eigenvector of $\wt{P}$.
    \begin{align*}
        \wt{P}\wt{v}[(F,f,c)] &= \sum_{(F',f',c')\in\States(\duwalksp)} \duwalksp[(F',f',c')\to(F,f,c)]\wt{v}[F',f',c']\\
        &= \sum_{\substack{(F',f',c')\in\States(\duwalksp)\\ (F',f')~\text{$0$-offset}}}\duwalksp[(F',f',c')\to(F,f,c)]\frac{{v}[F',f']}{T} +\\
        &\sum_{\substack{(F',f',c')\in\States(\duwalksp)\\ (F',f')~\text{not $0$-offset}}} \duwalksp[(F',f',c')\to(F,f,c)]v[F',f']\\
    \end{align*}
    If $(F,f,c)$ is a $0$-offset face, then the above quantity is equal to
    \begin{align*}
        \sum_{\substack{(F',f')\in\kface(\finalHDX)\\ (F',f')~\text{$0$-offset}}} \frac{\duwalk[(F',f')\to(F,f)]}{T}\cdot &\frac{v[F',f']}{T}\cdot T + \sum_{\substack{(F',f')\in\kface(\finalHDX)\\ (F',f')~\text{not $0$-offset}}}\frac{\duwalk[(F',f')\to(F,f)]}{T}v[(F',f')]\\
        =\ &\frac{1}{T}\sum_{(F',f')\in\kface(\finalHDX)}\duwalk[(F',f')\to(F,f)]v[(F',f')]\\
        =\ &\frac{1}{T}\lambda v[(F,f)]\\
        =\ &\lambda\wt{v}[(F,f,c)].
    \end{align*}
    And if $(F,f,c)$ is not a $0$-offset face, then the quantity is equal to
    \begin{align*}
        \sum_{\substack{(F',f')\in\kface(\finalHDX)\\ (F',f')~\text{$0$-offset}}} \duwalk[(F',f')\to(F,f)]\cdot &\frac{v[F',f']}{T}\cdot T + \sum_{\substack{(F',f')\in\kface(\finalHDX)\\ (F',f')~\text{not $0$-offset}}}\duwalk[(F',f')\to(F,f)]v[(F',f')]\\
        =\ &\frac{1}{T}\sum_{(F',f')\in\kface(\finalHDX)}\duwalk[(F',f')\to(F,f)]v[(F',f')]\\
        =\ &\lambda v[(F,f)]\\
        =\ &\lambda\wt{v}[(F,f,c)].
    \end{align*}
    Since for every right eigenvector $v$ of $P$, we can exhibit a right eigenvector $\wt{v}$ of $\wt{P}$, we can conclude that $\spec\left(\duwalk\right)\subseteq\spec\left(\duwalksp\right)$. 

\section{Spectral Gap of Outer Projection Chain} \label{app:oproj-appendix}

\subsection{Transition Probabilities of Outer Projection Chain}
\label{app:oproj-tp}
    
The table below summarizes the types of transition probabilities that occur between $i$ and $j$ in $\oproj$. Each row corresponds to a specific vertex of ``Source'' type, and provides (1) the transition probability to a specific vertex of ``Target'' type (where ``Same $k$-face'' denotes a transition from $(F, f)$ to $(F, f')$), and (2) the number of such transitions that occur from the source.  

\begin{center}
\begin{tabular}{|l|l|l|l|l|}
\hline
Source                      & Target                      & Same $k$-face & Probability                                                  & Count in $\Omega_j$, $j \neq i$, $(i, j) \in E(G)$ \\ \hline
\multirow{4}{*}{$0$-offset} & \multirow{2}{*}{$0$-offset} & Yes           & $\frac{1}{s - k} \cdot \frac{w_J}{DT}$                       & $1$                             \\ \cline{3-5} 
                            &                             & No            & $\frac{1}{k + 1} \cdot \frac{1}{s - k} \cdot \frac{w_J}{DT}$ & $(k + 1)(s - (k + 1))$          \\ \cline{2-5} 
                            & \multirow{2}{*}{$1$-offset} & Yes           & $\frac{1}{k + 1} \cdot \frac{1}{s - k} \cdot \frac{w_I}{D}$  & $(k + 1)$                       \\ \cline{3-5} 
                            &                             & No            & $\frac{1}{k + 1} \cdot \frac{1}{s - k} \cdot \frac{w_I}{D}$  & $(k + 1)(s - (k + 1))$          \\ \hline
\multirow{4}{*}{$1$-offset} & \multirow{2}{*}{$0$-offset} & Yes           & $\frac{1}{k + 1} \cdot \frac{1}{s - k} \cdot \frac{w_J}{DT}$ & $1$                             \\ \cline{3-5} 
                            &                             & No            & $\frac{1}{k + 1} \cdot \frac{1}{s - k} \cdot \frac{w_J}{DT}$ & $s - (k + 1)$                   \\ \cline{2-5} 
                            & \multirow{2}{*}{$1$-offset} & Yes           & $\frac{1}{k + 1} \cdot \frac{1}{s - k} \cdot \frac{w_I}{D}$  & $1$                             \\ \cline{3-5} 
                            &                             & No            & $\frac{1}{k + 1} \cdot \frac{1}{s - k} \cdot \frac{w_I}{D}$  & $s - (k + 1)$                   \\ \hline
\end{tabular}
\end{center}

\noindent Using the table, \pref{lem:stationary-main-lemma}, and the definition of projection chain from $\cite{jerrum}$, the transition probabilities of $\oproj$ are:
\begin{equation*}
    \oproj[i\to j] =
    \begin{cases}
        \displaystyle\frac{1}{2T} \cdot \frac{Tw_I+w_J}{[(2^k-1)Tw_I+w_J]},\quad &i \neq j,\text{ and } (i,j)\in E(\oproj), \\
        \displaystyle 1-\left(\frac{T-1}{T}\right)\cdot\left(\frac{Tw_I+w_J}{(2^k-1)Tw_I+w_J}\right), &i = j,\\
        0 &\text{otherwise.}
    \end{cases}
\end{equation*}

\subsection{Proof of \pref{lem:oproj-spectral-gap}}
\label{app:proof-of-oproj-spectral-gap}

Due to the symmetry of the transition probabilities and the partition $\Omega$, the spectrum of of $\oproj$ is easily computed from the spectrum of the following graph $L$:
\begin{itemize}
    \item $V(L) = [m]$,
    
    \item $E(L) = \{(i,j)\in E(\oproj) \mid i\neq j\}$.
\end{itemize}

\begin{observation}\label{obs:projo-is-line-graph}
$L$ is the line graph of the base expander $\baseExp$.
\end{observation}

\begin{proof}
By definition of the partition $\Omega$, there is a natural bijection between vertices in $V(L)$ and edges in $E(\baseExp)$. By construction, $(i,j)\in E(L)$ if and only if there exists $\{(F,f,e_i), (G,g,e_j)\}\in E(\duwalksp)$ such that $(F,f,e_i)\in \Omega_i \text{ and } (G,g,e_j)\in\Omega_j$. In the chain $\duwalksp$, two states $(F,f,e_i)$ and $(G,g,e_j)$ from different partition sets are connected only if they share a common endpoint in $\baseExp$.  Thus, $\{i,j\}\in E(\oproj)$ only if $e_i, e_j$ are adjacent in $\baseExp$. The if direction is straightforward from the construction of $\oproj$. So $L$ is the line graph of $\baseExp$. 
\end{proof}

\noindent The relationship between $\Spec(L(G))$ and $\Spec(G)$ is also well understood. 

\begin{theorem}[{\cite{sachs}}]\label{thm:Sachs67}
    If $G$ is a graph of degree $d$ with $n$ vertices and $L(G)$ its line graph, then the characteristic polynomials $\chi(G,\lambda)$ and $\chi(L(G),\lambda)$ satisfy 
    \begin{equation*}
        \chi(L(G),\lambda) = (\lambda+2)^{n(\frac{d}{2} - 1)}\chi(G, \lambda + 2 - d).
    \end{equation*}
\end{theorem}

\begin{proof}[Proof of \pref{lem:oproj-spectral-gap}] 
Using \pref{obs:projo-is-line-graph} and \pref{thm:Sachs67}, we relate the spectrum of $L$ to the spectrum of $\baseExp$. Specifically, if $\lambda$ is an eigenvalue of the normalized adjacency matrix of $\baseExp$, then $\frac{\lambda\dT+ \dT - 2}{2\dT-2}$ is an eigenvalue of the normalized adjacency matrix of $L$.  From this, one can deduce that $\Expansion(L) = \frac{\dT}{2\dT-2}\cdot\Expansion(G)$.
\begin{equation*}
    \oproj = \left(1 - \left(\frac{T-1}{T}\right)\cdot\frac{(w_J+Tw_I)}{\left(w_J+(2^{k}-1)Tw_I\right) }\right)\cdot I  + \frac{\dT-1}{\dT}\cdot\frac{(w_J+Tw_I)}{w_J+(2^{k}-1)Tw_I} \cdot \Adj{L}
\end{equation*}

It follows that if $v, \lambda$ is an eigenvector, eigenvalue pair of $\Adj{L}$, then
\[
    v, 1 - \left(\frac{T-1}{T}\right)\cdot\frac{(w_J+Tw_I)}{\left(w_J+(2^{k}-1)Tw_I\right) } + \lambda\cdot \frac{\dT-1}{\dT}\cdot\frac{(w_J+Tw_I)}{w_J+(2^{k}-1)Tw_I}
\]
is an eigenvector, eigenvalue pair of $\oproj$. Therefore,
\begin{align*}
    \Expansion(\oproj) &= \Expansion(L)\cdot\frac{T-1}{T}\cdot\frac{w_J+\dT w_I}{w_J+(2^{k}-1)\dT w_I}\\
    &=
    \frac{\Expansion(\baseExp)}{2} \cdot \frac{w_J+\dT w_I}{w_J+(2^{k}-1)\dT w_I}.  \qedhere
\end{align*}
\end{proof}

\section{Outer Restriction Chains} 
\label{app:oproj-appendix}

\subsection{Proof of \pref{lem:orest-iso}}
\label{app:proof-of-orest-iso}
Let $e_i,e_j\in E(\baseExp)$ be the edges corresponding to $\Omega_i,\Omega_j$ respectively. Suppose $e_i = \{u_i, v_i\}$ and $e_j = \{u_j, v_j\}$. Define a map $t_{ij}\colon e_i\rightarrow e_j$ to be $t_{ij}(u_i) = u_j, t_{ij}(v_i) = v_j$. Then $\orest{i}$ and $\orest{j}$ are isomorphic under the map $M_{ij}\colon \Omega_i\rightarrow \Omega_j, (F,f,e_i)\to (F, t_{ij}\circ f,e_j)$.

\subsection{Transition Probabilities of Outer Restriction Chains}
\label{app:orest-tp}
Since the restriction chains are isomorphic, we can focus on $\orest{1}$ without loss of generality. Using the decomposition rule given in \pref{sec:decompose}, we can compute the transition probabilities of $\orest{1}$:
\begin{itemize}
    \item For all 0-offset $(F, f, e_1)$, the self loop probability is $$\frac{\dT - 1}{\dT}+ \frac{w_J}{D\dT(s - k)}.$$ The transition probability to each of its $(k + 1)(s - k -1)$ adjacent 0-offset neighbors $(F', f, e_1)$ is $$\frac{w_J}{D\dT(k + 1)(s - k)}.$$ The transition probability to each of its $(k + 1)(s - k)$ non-0-offset neighbors $(F', f')$ is $$\frac{w_I}{D(k + 1)(s - k)}.$$

    \item For all 1-offset $(F, f, e_1)$, the self loop probability is
    \[
        \frac{(T - 1)}{T(k + 1)} + \frac{w_I}{D(k + 1)(s - k)}+\frac{k}{2(k+1)(s - k)}.
    \]
    The transition probability to each of its $(s - k)$ 0-offset neighbors $(F', f')$ is $$\frac{w_J}{D\dT(k + 1)(s - k)}.$$ The transition probability to each of its $k$ non-0-offset neighbors with identical base $k$-face $(F, f', e_1)$ is $$\frac{1}{2(k + 1)(s - k)}.$$ The transition probability to each of its $(s - k - 1)$ non-0-offset neighbors with a different base $k$-face $(F', f', e_1)$ reached by deleting the lonely\footnote{Recall that ``lonely'' was defined in \pref{def:lonely}} vertex and adding back a different lonely vertex is $$\frac{w_I}{D(k + 1)(s - k)}.$$  The transition probability to each of its $2k(s - k - 1)$ non-0-offset neighbors with a different base $k$-face $(F', f', e_1)$ reached by deleting a non-lonely vertex and adding back any other vertex is $$\frac{1}{2(k + 1)(s - k)}.$$

    \item For the remaining $(F, f, e_1)$, the self loop probability is $$\frac{1}{2(s - k)}.$$ The transition probability to each of its $(k + 1)$ neighbors with an identical base $k$-face $(F, f', e_1)$ is $$\frac{1}{2(k + 1)(s - k)}.$$ The transition probability to each of its $2(k + 1)(s - k - 1)$ neighbors with a different base $k$-face $(F',f',e_1)$ is also $$\frac{1}{2(k + 1)(s - k)}.$$
\end{itemize}

\section{Spectral Gap of Inner Projection Chain}   \label{app:inner-proj}

\subsection{Lazy Random Walks}

\noindent Both the inner projection and the inner restriction chains have self-loops, so it will be useful to first present some preliminary results on lazy random walks. If we start with Markov chain $\wt{M} = (\Omega, \wt{P})$ and wish to add a uniform self loop probability to each state to get Markov chain $M = (\Omega, P)$, we write $P$ as a convex combination of $\wt{P}$ and $I$:
$$
P = c \cdot I + (1 - c) \cdot \wt{P}, \text{ where } 0 \leq c \leq 1
$$
Since this convex combination will appear a few different times throughout this paper, we'll prove a basic fact about the spectral gap of $P$:
\begin{lemma} \label{lem:scaling}
    For $M = (\Omega, P)$ as defined above:
    $$
        \gap(M) = (1 - c) \cdot \gap(\wt{M})
    $$
\end{lemma}
\begin{proof}
    Let $\lambda$ be any eigenvalue of $\wt{P}$, with associated eigenvector $v$. Then, $v$ is also an eigenvector for $P$ for eigenvalue:
    $$
        c + (1 - c) \cdot \lambda(M)
    $$
    To see this, $Pv = \left[c \cdot I + (1 - c) \cdot \tilde{P}\right]v = cv + (1 - c)\left(\tilde{P}v\right) = \left[c + (1 - c) \cdot \lambda\right]v$.  The spectrum of $\tilde{M}$ is a linear shift and scaling of the spectrum of $M$, and the spectral gap scales by $(1 - c)$.
\end{proof}

\subsection{Transition Probabilities of Outer Projection Chain}\label{app:inner-proj-prob}
The table below indicates the transition probabilities from a specific face of ``Source'' type in $F_i$ to various ``Target'' faces in $F_j$ for $j \neq i$. In the last column, we count transitions to \emph{any} $F_j$, rather than a specific $F_j$; this made our computations much easier. Due to the symmetry of the $\{F_i\}$ partition elements, to get the transition from $F_i$ to a specific $F_j$, we simply divide the transition probability to $\bigcup_{j \neq i} F_j$ by the number of $F_j$ adjacent to $F_i$, which is $(k + 1)(s - (k + 1))$.

\begin{center}
\begin{tabular}{|l|l|l|l|l|}
\hline
Source                      & Delete                    & Target           & Probability                                                  & Count in $F_j$, $j \neq i$ \\ \hline
\multirow{2}{*}{$0$-offset} & \multirow{2}{*}{anything} & $0$-offset       & $\frac{1}{k + 1} \cdot \frac{1}{s - k} \cdot \frac{w_J}{DT}$ & $1$            \\ \cline{3-5} 
                            &                           & $1$-offset       & $\frac{1}{k + 1} \cdot \frac{1}{s - k} \cdot \frac{w_I}{D}$  & $1$             \\ \hline
\multirow{4}{*}{$1$-offset} & \multirow{2}{*}{minority}   & $0$-offset       & $\frac{1}{k + 1} \cdot \frac{1}{s - k} \cdot \frac{w_J}{DT}$ & $1$                     \\ \cline{3-5} 
                            &                           & $1$-offset       & $\frac{1}{k + 1} \cdot \frac{1}{s - k} \cdot \frac{w_I}{D}$  & $1$                     \\ \cline{2-5} 
                            & \multirow{2}{*}{majority}     & $1$-offset       & $\frac{1}{k + 1} \cdot \frac{1}{s - k} \cdot \frac{1}{2}$    & $1$                  \\ \cline{3-5} 
                            &                           & $2$-offset       & $\frac{1}{k + 1} \cdot \frac{1}{s - k} \cdot \frac{1}{2}$    & $1$                  \\ \hline
\multirow{4}{*}{$t$-offset} & \multirow{2}{*}{minority} & $t$-offset       & $\frac{1}{k + 1} \cdot \frac{1}{s - k} \cdot \frac{1}{2}$    & $1$                  \\ \cline{3-5} 
                            &                           & $(t - 1)$-offset & $\frac{1}{k + 1} \cdot \frac{1}{s - k} \cdot \frac{1}{2}$    & $1$                  \\ \cline{2-5} 
                            & \multirow{2}{*}{majority}     & $t$-offset       & $\frac{1}{k + 1} \cdot \frac{1}{s - k} \cdot \frac{1}{2}$    & $1$        \\ \cline{3-5} 
                            &                           & $(t + 1)$-offset & $\frac{1}{k + 1} \cdot \frac{1}{s - k} \cdot \frac{1}{2}$    & $1$        \\ \hline
\end{tabular}
\end{center}

\noindent Using the table above, \pref{lem:stationary-orest}, and the framework of \cite{jerrum}, the specific transition probabilities for each state in the projection chain are:
\begin{itemize}
    \item $\displaystyle p := \frac{1}{\dT(k + 1)(s - k)}\cdot\frac{[(2^k - 2)T + 1]Tw_I + w_J}{(2^k - 1)\dT w_I + w_J}$ to each of its $(k + 1)(s - (k + 1))$ neighbors.
    
    \item $1 - (k + 1)(s - (k + 1))p$ for self loops, which can be verified to be nonzero.
\end{itemize}

\subsection{Proof of \pref{lem:iproj-spectral-gap}}\label{app:inner-proj-gap}
Let $\wt{\duwalkbase}$ be the non-lazy version (i.e. no self loops) of $\duwalkbase$.  Since in our construction, $\duwalkbase$ is a complete complex, $w(F)$ is uniform over $F \in \Face{k}$, so all transitions in $\wt{\duwalkbase}$ are also uniform. To understand the spectrum of $\iproj$, we can express the transition matrix of $\iproj$ as: 
\[
    (k + 1)(s - (k + 1))p \cdot \wt{\duwalkbase} + [1 - (k + 1)(s - (k + 1))p] \cdot \Id
\]
Luckily, for $\duwalkbase$ a complete complex, the spectrum of $\wt{\duwalkbase}$ is well understood.  The following can be deduced from the main theorem of \cite{kaufman-oppenheim}.

\begin{theorem} \label{gap-up-down}
    $\displaystyle\gap(\duwalkbase) \geq \frac{1}{(k + 1)} \enspace.$
\end{theorem}

We can now compute the second largest eigenvalue of the non-lazy walk $\wt{\duwalkbase}$.

\begin{corollary}
$\displaystyle\gap(\wt{\duwalkbase}) \geq \frac{1}{(k + 1)} + \frac{1}{(k + 1)(s - k - 1)} \enspace.$
\end{corollary}

\begin{proof}
Since we are working with a complete complex, all weights on sets of a given size are uniform. Thus, the self-loop probability of $\duwalkbase$ is $\frac{1}{s - k}$. 

We can next write $\duwalkbase = \frac{1}{s - k} \cdot I + (1 - \frac{1}{s - k}) \cdot \wt{\duwalkbase}$. Using \pref{lem:scaling}, we conclude that 
$$
    \gap(\wt{\duwalkbase}) = \frac{\gap(\duwalkbase)}{(1 - \frac{1}{s - k})}
$$
We get the desired result after substituting $\frac{1}{(k + 1)}$ as a lower bound for $\gap(\wt{\duwalkbase})$.
\end{proof}

\begin{proof}[Proof of \pref{lem:iproj-spectral-gap}]
By \pref{lem:scaling} again, we have $$\gap(\iproj) = \gap(\wt{\duwalkbase}) \cdot (k + 1)(s - (k + 1))p.$$ Substituting for $p$:
\begin{align*}
   \gap(\iproj) &\geq \left[\frac{1}{(k + 1)} + \frac{1}{(k + 1)(s - k - 1)}\right] \cdot \left[\frac{s - k - 1}{\dT(s - k)}\cdot\frac{[(2^k - 2)T + 1]Tw_I + w_J}{(2^k - 1)\dT w_I + w_J} \right] \\
   &\geq \left[\frac{1}{(k + 1)} + \frac{1}{(k + 1)(s - k - 1)}\right] \cdot \frac{1}{2T} \\
   &\geq \frac{1}{2T(k + 1)}
\end{align*}
It can be verified that $\displaystyle \frac{1}{2T}$ is a lower bound on $\displaystyle \frac{s - k - 1}{\dT(s - k)}\cdot\frac{[(2^k - 2)T + 1]Tw_I + w_J}{(2^k - 1)\dT w_I + w_J}$.
\end{proof}

\section{Spectral Gap of Inner Restriction Chain} \label{app:irest-appendix}

\subsection{Transition Probabilities of Inner Restriction chain}\label{app:irest-prob}
The transition probabilities can be summarized succinctly:

\begin{center}
\begin{tabular}{|l|l|l|l|}
\hline
Source                      & Delete   & Target           & Probability                                                  \\ \hline
$0$-offset                  & anything & $1$-offset       & $\frac{1}{k + 1} \cdot \frac{1}{s - k} \cdot \frac{w_I}{D}$  \\ \hline
\multirow{2}{*}{$1$-offset} & minority   & $0$-offset       & $\frac{1}{k + 1} \cdot \frac{1}{s - k} \cdot \frac{w_J}{DT}$ \\ \cline{2-4} 
                            & majority     & $2$-offset       & $\frac{1}{k + 1} \cdot \frac{1}{s - k} \cdot \frac{1}{2}$    \\ \hline
\multirow{2}{*}{$t$-offset} & minority & $(t - 1)$-offset & $\frac{1}{k + 1} \cdot \frac{1}{s - k} \cdot \frac{1}{2}$    \\ \cline{2-4} 
                            & majority     & $(t + 1)$-offset & $\frac{1}{k + 1} \cdot \frac{1}{s - k} \cdot \frac{1}{2}$    \\ \hline
\end{tabular}
\end{center}

\subsection{Proof of \pref{lem:irest-spectral-gap}}  \label{app:irest-gap}
\noindent We can also define a related chain $U$, that has the same state space and transitions as $R$, but the self loop probabilities are uniform across all vertices. More precisely:
\begin{itemize}
    \item For \emph{all} hypercube vertices, the self loop probability is $\displaystyle 1 - \frac{w_I}{D(S - k)}$. The transition probability to each of their $(k + 1)$ neighbors in the hypercube is $$\frac{w_I}{D(k + 1)(s - k)}.$$ 
\end{itemize}
The goal of this section is to bound on the spectral gap of $\irest$. Our approach relates the spectrum of $R$ to the spectrum of $U$. Due to the uniformity of the self loop probabilities, the spectrum of $U$ is easy to compute. 

\noindent For ease, we will write $\wt{D} = 2w_J + w_I\dT(2^{k + 1} - 2)$. The key bound on $\gap(\irest)$, which we provide a proof of in \pref{app:spec-gap-key-bound} is the following:
\begin{lemma} \label{lem:restriction_compare_gaps}
    $$\gap(\irest) \geq \frac{w_J \wt{D}}{2^{k + 1} \cdot (\dT w_I)^2} \cdot \gap(U) \geq \frac{w_J}{2 \dT w_I} \cdot \gap(U) \enspace.$$ 
\end{lemma}
\noindent We are able to explicitly compute $\gap(U) = \frac{2w_I}{D(k + 1)(s - k)}$ (see \pref{lem:hypercube-gap}).  The conclusion of \pref{lem:irest-spectral-gap} is then immediate.

\subsection{Variational Characterization of Spectral Gap}   \label{app:var-char-gap}

\noindent We will also use a different, variational characterization of the spectral gap of a time-reversible Markov chain $(\Omega, P)$, which will prove useful when working with self loops that have different probabilities. This characterization provides bounds on $\lambda_2$ without forcing us to analyze the chain's entire spectrum \cite{poincare}.

\begin{definition}
Let $M = (\Omega, P)$ be a time-reversible Markov chain. For functions $f, g: \Omega \to \mathbb{R}$, the \emph{Dirichlet form} corresponding to $M$ is:
$$
\dir_M(f, g) = \frac{1}{2} \sum_{x \in \Omega} \sum_{y \in \Omega} \pi_M(x) M[x \to y] \cdot [f(x) - f(y)][g(x) - g(y)]
$$
We may omit the subscript $M$ when there is no ambiguity. 
\end{definition}

\begin{definition}
Again, let $(\Omega, P)$ be a time-reversible Markov chain with stationary distribution $\pi_M$. For a functions $f: \Omega \to \mathbb{R}$, the \emph{variance} corresponding to $M$ is:
$$
\Var_M(f) = \frac{1}{2} \sum_{x \in \Omega} \sum_{y \in \Omega} \pi_M(x) \pi_M(y) \cdot [f(x) - f(y)]^2
$$
We may omit the subscript $M$ when there is no ambiguity. This definition is equivalent to
$$
\Var_M(f) =\E_{\pi_M}[f^2] - \E_{\pi_M}[f]^2
$$
These definitions are equivalent because for $X, Y$ i.i.d, $\Var(X) = \frac{1}{2} \E[(X - Y)^2]$.
\end{definition}

\begin{theorem} \label{thm:variational}
    Let $M = (\Omega, P)$ be a time-reversible Markov Chain. Then:
    $$
    \gap(M) = \inf\left\{ \frac{\dir_M(f, f)}{\Var_M(f)} \; \large| \; f: \Omega \to \mathbb{R}, \Var_M(f) \neq 0 \right\}
    $$
\end{theorem}

Often, we will not be able to compute the exact spectral gap of a chain, but it will suffice to have a lower bound on it. We can determine whether $\lambda$ is a lower bound on $\gap(M)$ by checking if it satisfies the \emph{Poincar\'e inequality}:

\begin{definition}
    We say $\lambda \geq 0$ satisfies the \emph{Poincar\'e inequality} if for all $f: \Omega \to \mathbb{R}$:
    $$
    \lambda \cdot \Var_M(f) \leq \dir_M(f, f)
    $$
    By the variational characterization of spectral gap, we would also have $\lambda \leq \gap(M)$. 
\end{definition}

\subsection{Proof of \pref{lem:restriction_compare_gaps}}    \label{app:spec-gap-key-bound}

\begin{lemma} \label{lem:hypercube-gap}
    Let $H$ be the uniform, non-lazy walk on the $(k + 1)$-dim. hypercube. Then $\lambda_2(H) = \frac{2}{k + 1}$. 
\end{lemma}
\begin{proof} 
See \cite{babai1979spectra} for a thorough treatment of Cayley graphs. The $(k + 1)$-dimensional hypercube is the Cayley graph derived from the cyclic group $\mathbb{Z}_2^{k + 1}$.
\end{proof}

\begin{observation}
    $\lambda(U)$ is $\frac{2w_I}{D(k + 1)(S - k)}$. 
\end{observation}

\begin{proof}
Let $P_H$ denote the transition matrix of a uniform random walk on a $(k + 1)$-dimensional hypercube, with no self loops. Then, the transition matrix $P_U$ of $U$ can be expressed as: 
$$ P_U = \left(1 - \frac{w_I}{D(S - k)}\right) \cdot I + \frac{w_I}{D(S - k)} \cdot P_H $$
By \pref{lem:scaling}, we have $\lambda(U) = \lambda(H) \cdot \frac{w_I}{D(S - k)}$, so we get the desired result via \pref{lem:hypercube-gap}.
\end{proof}

\noindent We also observe that the stationary distribution of $U$, which we will call $\pi_U$, is uniform over the $2^{k + 1}$ states. The stationary distribution of $\irest$, denoted $\pi_{\irest}$ can also be described explicitly. 

\begin{observation}
The stationary distribution $\pi_{\irest}$ of chain $\irest$ is
\begin{equation*}
    \pi_{\irest}(x) = 
    \begin{cases}
        \frac{w_J}{2w_J + w_I\dT(2^{k + 1} - 2)} &\text{ if } x \in \{\vec{0}, \vec{1}\} \\
        \frac{\dT w_I}{2w_J + \dT w_I(2^{k + 1} - 2)} &\text{ otherwise}
    \end{cases}
\end{equation*}
\end{observation}
\begin{proof}
By time reversibility of $\irest$ \cite{jerrum}, the detailed balance equations imply that for all $y$ that are $t$-offset, for $t \geq 1$, the stationary probability $\pi_{\irest}(y)$ is the same, and $\pi_{\irest}(\vec{0}) = \pi_{\irest}(\vec{1})$. 

Let $x$ be $0$-offset and $y$ be $1$-offset. Again, by time-reversibility of $R_i$ and detailed balance:
$$
\pi_{\irest}(x) \cdot \frac{w_I}{D(k + 1)(S - k)} = \pi_{\irest}(y) \cdot \frac{w_J}{D\dT(k + 1)(S - k)}
$$
This tells us $\pi_{\irest}(x) = \frac{w_J}{\dT w_I} \cdot \pi_{\irest}(y)$. Solving for $\sum_{x \in \{0, 1\}^{k + 1}} \pi_{\irest}(x) = 1$ gives the desired result.  

\end{proof}
\noindent Recall that we write $\Tilde{D} = 2w_J + w_I\dT(2^{k + 1} - 2)$.

\begin{proof}[Proof of \pref{lem:restriction_compare_gaps}] Let $g$ be a real-valued function over the $k$-faces of $\LD(G,\SimpComp)$. Using \pref{thm:variational}, it suffices to prove that for all $g$,
$$
\frac{\dir_{\irest}(g, g)}{\Var_{\irest}(g, g)} \geq  \frac{w_J \Tilde{D}}{2^{k + 1} \cdot (\dT w_I)^2} \cdot \frac{\dir_U(g, g)}{\Var_U(g, g)}
$$
First, we compute both $\dir_U(g, g)$ and $\dir_R(g, g)$. 
\begin{align*}
    \dir_U(g, g) &= \frac{1}{2} \sum_{x, y \in \{0, 1\}^{(k + 1)}} \pi_U(x) \cdot [g(x) - g(y)]^2 \cdot P_U(x, y) \\ 
    &= \frac{1}{2} \cdot \frac{1}{2^{(k + 1)}} \cdot \frac{w_I}{D(k + 1)(S - k)} \sum_{x, y \in \{0, 1\}^{(k + 1)}} [g(x) - g(y)]^2
\end{align*}
\begin{align*}
    \dir_{\irest}(g, g) &= \frac{1}{2} \sum_{x, y \in \{0, 1\}^{(k + 1)}: x \in \{\vec{0}, \vec{1}\}} \pi_{\irest}(x) \cdot [g(x) - g(y)]^2 \cdot P_{\irest}(x, y) \\
    & \; \; \; \; \; \; \; \; \; \; \; \; + \frac{1}{2} \sum_{x, y \in \{0, 1\}^{(k + 1)}: x \notin \{\vec{0}, \vec{1}\}} \pi_{\irest}(x) \cdot [g(x) - g(y)]^2 \cdot P_{\irest}(x, y) \\
    &= \frac{1}{2} \sum_{x, y \in \{0, 1\}^{(k + 1)}: x \in \{\vec{0}, \vec{1}\}} \frac{w_J}{\Tilde{D}} \cdot [g(x) - g(y)]^2 \cdot \frac{w_I}{D(k + 1)(S - k)} \\
    & \; \; \; \; \; \; + \frac{1}{2} \sum_{x \in \{0, 1\}^{(k + 1)}: x\text{ 1-balanced}} \frac{\dT w_I}{\Tilde{D}} \cdot \left( \sum_{y \in \{0, 1\}^{(k + 1)}: y\text{ 0-balanced}} [g(x) - g(y)]^2 \cdot \frac{w_J}{D\dT(k + 1)(S - k)} \right.\\
    & \; \; \; \; \; \; \; \; \; \; \; \; \; \; \; \; \left. +\sum_{y \in \{0, 1\}^{(k + 1)}: y\text{ not 0-balanced}} [g(x) - g(y)]^2 \cdot \frac{1}{2(k + 1)(S - k)} \right) \\
    & \; \; \; \; \; \; + \frac{1}{2} \sum_{x, y \in \{0, 1\}^{(k + 1)}: x,y \notin \{\vec{0}, \vec{1}\}} \frac{\dT w_I}{\Tilde{D}} \cdot [g(x) - g(y)]^2 \cdot \frac{1}{2(k + 1)(S - k)} \\
    &\geq \frac{1}{2} \cdot \frac{w_Iw_J}{\Tilde{D}D(k + 1)(S - k)} \sum_{x, y \in \{0, 1\}^{(k + 1)}} [g(x) - g(y)]^2
\end{align*}
From the above computations, we can conclude that
$$
\dir_{\irest}(g, g) \geq \frac{2^{(k + 1)} \cdot w_J}{\Tilde{D}} \cdot \dir_U(g, g)
$$
Similarly, we can compute both $\Var_U(g)$ and $\Var_R(g)$:
\begin{align*}
    \Var_U(g) &= \frac{1}{2} \sum_{x, y \in \{0, 1\}^{k + 1}} \pi_U(x) \pi_U(y)[f(x) - f(y)]^2 \\
    &= \frac{1}{2} \cdot \frac{1}{2^{2(k + 1)}} \sum_{x, y \in \{0, 1\}^{k + 1}} [f(x) - f(y)]^2 \\
    \Var_{\irest}(g, g) &= \frac{1}{2} \sum_{x, y \in \{\vec{0}, \vec{1}\}} \pi_{\irest}(x)\pi_{\irest}(y) [f(x) - f(y)]^2 \\
    & \; \; \; \; \; \; \; \; \; \; \; \; + \frac{1}{2} \sum_{\substack{x \in \{\vec{0}, \vec{1}\}, \; y \in \{0, 1\}^{k + 1} \setminus \{\vec{0}, \vec{1}\} \text{ or } \\ x \in \{0, 1\}^{k + 1} \setminus \{\vec{0}, \vec{1}\}, \; y \in  \{\vec{0}, \vec{1}\}}} \pi_{\irest}(x)\pi_{\irest}(y) [f(x) - f(y)]^2 \\
    & \; \; \; \; \; \; \; \; \; \; \; \; + \frac{1}{2} \sum_{x, y \in \{\vec{0}, \vec{1}\}} \pi_{\irest}(x)\pi_{\irest}(y) [f(x) - f(y)]^2 \\
    &= \frac{1}{2} \sum_{x, y \in \{\vec{0}, \vec{1}\}} \frac{w_J^2}{\Tilde{D}^2} [f(x) - f(y)]^2 + \frac{1}{2} \sum_{\substack{x \in \{\vec{0}, \vec{1}\}, \; y \in \{0, 1\}^{k + 1} \setminus \{\vec{0}, \vec{1}\} \text{ or } \\ x \in \{0, 1\}^{k + 1} \setminus \{\vec{0}, \vec{1}\}, \; y \in  \{\vec{0}, \vec{1}\}}} \frac{\dT w_I w_J}{\Tilde{D}^2} [f(x) - f(y)]^2 \\
    & \; \; \; \; \; \; \; \; \; \; \; \; + \frac{1}{2} \sum_{x, y \in \{\vec{0}, \vec{1}\}} \frac{(\dT w_I)^2}{\Tilde{D}^2} [f(x) - f(y)]^2 \\
    &\leq \frac{1}{2} \cdot \frac{(\dT w_I)^2}{\Tilde{D}^2} \sum_{x, y \in \{0, 1\}^{k + 1}} [f(x) - f(y)]^2
\end{align*}
From the above computations, we can conclude that
$$
\Var_{\irest}(g) \leq \frac{2^{2(k + 1)} \cdot (\dT w_I)^2}{\Tilde{D}^2} \Var_U(g)
$$
Combining this with what we know about $\dir_{\irest}(g, g)$ and $\dir_U(g, g)$, we conclude the lemma. 
\end{proof}

\end{document}